%% file: gkc_coalescent.tex
\documentclass{article}                     
\usepackage{graphicx, natbib}

\usepackage{mathptmx}      
\usepackage{amssymb, amsmath, verbatim, amsfonts, amsthm}

\newtheorem{proposition}{Proposition}

\newtheorem{theorem}{Theorem}
\newtheorem{lemma}{Lemma}

\bibpunct{[}{]}{;}{n}{}{,}

\newcommand{\sumkkjj}{\sum_{k,k'=1}^d \sum_{j,j'=1}^n}
\newcommand{\Ikkjj}{I(x_{k,j} = x_{k',j'})}

\newcommand{\cphio}{\ | \ \phi_1 \ne 1}
\newcommand{\cRallB}{\ | \ R,\{B_i\}}

\newcommand{\p}{p_1}
\newcommand{\q}{p_2}
\newcommand{\RB}{\frac{R}{B}}
\newcommand{\RLB}{\frac{RL}{B}}
\newcommand{\LPLSL}{\lim_\text{LPLS}}
\renewcommand{\P}{\mathcal{P}}

\newcommand{\brand}{\hat{b}}
\newcommand{\Ih}{\mathcal{I}_{h}^\epsilon}
\newcommand{\cIh}{\ | \ \frac{RL}{B} \in \mathcal{I}_{h}^\epsilon}
\newcommand{\F}{\mathcal{F}}
\renewcommand{\P}{\mathcal{P}}
\newcommand{\PL}{\mathcal{P^\text{lab}}}
\newcommand{\D}{\mathcal{D}}
\newcommand{\KC}{\text{KC}}
\newcommand{\IM}{\text{IM}}
\renewcommand{\SS}{\text{SS}}
\newcommand{\G}{\mathcal{G}}
\newcommand{\GKC}{\text{G/KC}}
\newcommand{\Td}{T_\text{scat}}

\newcommand{\SR}{\Delta_\text{sample}}
\newcommand{\Torus}{\mathbb{T}^2}
\newcommand{\Dsample}{\Delta_\text{sample}}
\newcommand{\Emut}{E_\text{mut}}
\newcommand{\calA}{\mathcal{A}}
\newcommand{\SPL}{\bar{\Xi}}
\newcommand{\SP}{\Xi}
\newcommand{\srand}{\hat{s}}
\newcommand{\DNA}{\mathcal{S}}
\newcommand{\IS}{\text{WC}}

\newcommand{\gen}{\text{gen}}
\newcommand{\AV}{\text{ave}}
\begin{document}

\title{The Distribution of $F_{st}$ and other Genetic Statistics for a Class of Population Structure Models }


\author{Sivan Leviyang\thanks{Georgetown University.  Department of Mathematics.  sr286@georgetown.edu}}

\maketitle



\begin{abstract}
We examine genetic statistics used in the study of structured populations.  In a 1999 paper, Wakeley observed that the coalescent process associated with the finite island model can be decomposed into a scattering phase and a collecting phase.  In this paper, we introduce a class of population structure models, which we refer to as G/KC models, that obey such a decomposition.  In a large population, large sample limit we derive the distribution of the statistic $F_{st}$ for all G/KC models under the assumptions of strong or weak mutation.  We show that in the large population, large sample limit the island and two dimensional stepping stone models are members of the G/KC class of models, thereby deriving the distributions of $F_{st}$ for these two well known models as a special case of a general formula.  We show that our analysis of $F_{st}$ can be extended to an entire class of genetic statistics, and we use our approach to examine homozygosity measures.  Our analysis uses coalescent based methods.

\end{abstract}

\input{gkc_coalescent_body}

\vspace{1cm}

\begin{flushleft}
\textbf{Acknowledgements} \ \ 
I thank M. Hamilton for innumerable conversations about the current state and future direction of statistical testing in population genetics.
\end{flushleft}

\newcommand{\noopsort}[1]{} \newcommand{\printfirst}[2]{#1}
  \newcommand{\singleletter}[1]{#1} \newcommand{\switchargs}[2]{#2#1}


\end{document}

%% file: gkc_coalescent_body.tex




\section{Introduction}
\setcounter{equation}{0}
\setcounter{lemma}{0}

	Biological populations are often divided into subpopulations between which migration is restricted.   Such populations, referred to as structured populations, have been an important area of population genetics research since the 1930s \cite{Wright_1931_Genetics}.  In application, various statistics based on genetic data are used in hypothesis testing to understand structured populations.   An example of such a genetic statistic is $F_{st}$.  $F_{st}$, which we define precisely below, is used to test for the presence of population structure and to estimate migration rates \cite{Weir_1984_Evolution, Weir_2002_Ann_Rev_Genet, Wright_1942_Genetics}.  

	The analysis of $F_{st}$ has a long history that reflects the history of population genetics.  $F_{st}$ was introduced by Wright in the context of single locus, biallelic data \cite{Wright_1942_Genetics}.  Over time, $F_{st}$ was generalized to multiple loci, multiple allele data (e.g. \cite{Nei_1973_Proc_Nat_Acad_Sci, Weir_1984_Evolution}) and to sequence data  (e.g. \cite{Lynch_1990_Mol_Bio_Evol}).  Initially, Wright considered $F_{st}$ under the infinite island model for population structure.  Over time, $F_{st}$ was analyzed under the finite island model  (e.g. \cite{Nei_1977_TPB, Takahata_1983_Genetics}), stepping stone models (e.g. \cite{Crow_1984_PNAS}), and some more general population structure models (e.g. \cite{Wilkinson_Herbots_1998_J_Math_Bio}).  The method of analysis of $F_{st}$ moved from frequency based methods to coalescent methods (e.g. \cite{Cox_2002_Ann_Appl_Prob, Nei_1977_TPB, Slatkin_1991_Genet_Res_Camb}).
	
	But today, the distribution of $F_{st}$ is still poorly understood.    The distribution of $F_{st}$ is known only for the island model in the case of single locus, multiallelic data \cite{Leviyang_2008}.    How the distribution of $F_{st}$ changes under different models of population structure and genetic data is not known.    $F_{st}$, in all its forms, is just one example of a general problem.  We know very little about the distribution of genetic statistics under population structure, and what we know about these statistics is confined to very specific models.  In application, this lack of knowledge has important consequences.  First, since distributions are not known, the construction of confidence intervals can only be done through resampling techniques \cite{Weir_Book_Genetic_Data_Analysis_II}.  Second, since results are not generalizable beyond specific models, hypothesis tests assume a null hypothesis that includes a specific form of population structure.  By including such assumptions the utility of hypothesis testing is severely limited \cite{Whitlock_1999_Heredity}.

	In this paper we address some of these issues by analyzing $F_{st}$ and other genetic statistics over a class of population structure models which we call G/KC models.  G/KC models are limiting versions of models that obey the scattering-collecting phase decomposition introduced by Wakeley \cite{Wakeley_1999_Genetics}.  We consider a large population, large sample limit, thereby removing statistical variance and focusing on evolutionary variance (see \cite{Weir_Book_Genetic_Data_Analysis_II} for a discussion of this issue).  In this setting, we derive a formula for the distribution of $F_{st}$  for any G/KC model under the assumption of weak or strong mutation.  We show that in the large population, large sample limit, the island and two-dimensional stepping stone models correspond to certain G/KC models, thereby deriving the distribution of $F_{st}$ for both the island and stepping stone models as a special case of the more general formula for G/KC models.  We further show that our approach to the analysis of $F_{st}$ can be applied to a whole class of genetic statistics which we refer to as diversity measures and of which $F_{st}$ is an example.  In proving our results we assume a haploid population of constant size under a Moran mating scheme.
	
	Our analysis uses coalescent based methods, see \cite{Durrett_Book_Probability_Models} for a good introduction.  With this in mind, we describe the island, stepping stone, and G/KC models by specifying their coalescent processes.  We consider the island and two dimensional stepping stone models because of their central role in population genetics.  Other models can be analyzed by our methods, see \cite{Matsen_2006_Genetics} for a whole class of such models.
	
	The rest of this paper is organized as follows.  In section \ref{S:Basic_Defs} we introduce basic definitions that we need to present our results.  In section \ref{S:results} we present our results.  In section \ref{S:applications} we apply our results in several different settings of practical interest.  We discuss $F_{st}$ under a single locus, infinite allele model, under a mutilocus, biallelic model, and under an infinite sites model.  We also use our results to compare homozygosity measures under the island and stepping stone models.  Sections \ref{S:convergence}-\ref{S:weak} contain the proofs of the theorems stated in section \ref{S:results}.  Section \ref{S:convergence} connects the G/KC coalescent to the island and stepping stone model coalescents, while sections \ref{S:strong} and \ref{S:weak} prove results concerning $F_{st}$.

\section{Diversity Measures and Coalescent Models}  \label{S:Basic_Defs}
\setcounter{equation}{0}
\setcounter{lemma}{0}

	In this section we introduce some basic definitions.  In subsection \ref{S:DM} we give a general definition for diversity measures and the diversity measure $F_{st}$ in particular.  In subsection \ref{S:CM} we introduce the island and stepping stone model coalescent processes along with the Kingman coalescent.  Finally in subsection \ref{S:GKC} we introduce the G/KC coalescent.

\subsection{Diversity Measures}  \label{S:DM}
\setcounter{equation}{0}

	We consider a population that is separated into $D$ subpopulations.  We refer to these subpopulations as demes.  Each deme is composed of $N$ individuals and the population size of each deme is fixed at $N$ over all times.  At time $0$, we sample individuals from $d$ demes.  From each sampled deme we sample $n$ individuals.  So we sample $nd$ individuals in all.

	From each sampled individual we obtain a genetic state.  Let $\DNA$ be the set of all mappings from $\mathbb{N} \to [0,1]$.  A genetic state $\xi$ is an element of $\DNA$.   Set
\begin{equation}
x_{k,j}^\gen = \text{ genetic state of } j\text{th sampled individual in }k\text{th sampled deme}. 
\end{equation}
Then $x_{k,j}^\gen \in \DNA$ and $x_{k,j}^\gen(i) \in \{0,1\}$.  We say that $G$ is a diversity measure if it is a bounded function of $x_{k,j}^\gen$ over $k=1,2,\dots,d$ and $j=1,\dots,n$ that is symmetric in $j$ for fixed $k$ and symmetric in $k$ for fixed $j$.

	Let $\chi()$ be the indicator function (i.e. $\chi(\text{true}) = 1, \chi(\text{false}) = 0)$.   We introduce two specific diversity measures on which our technical analysis focuses:  the homozygosity measures $\phi_0$,$\phi_1$ and $F_{st}$.  We use the definition and notation given by Nei in \cite{Nei_1973_Proc_Nat_Acad_Sci}.

\begin{flushleft}
\textit{Homozygosity Measures:}\\
\end{flushleft}

\begin{equation}   \label{E:formula_phi_0_k}
\phi_{0,k} = \frac{1}{n^2} \sum_{j,j' = 1}^n \chi(x_{k,j}^\gen = x_{k,j'}^\gen).
\end{equation}

\begin{equation}   \label{E:formula_phi_0}
\phi_0 = \frac{1}{d} \sum_{k=1}^d \phi_{0,k}.
\end{equation}

\begin{equation}  \label{E:formula_phi_1}
\phi_1 = \frac{1}{d^2} \sum_{k,k'=1}^d \frac{1}{n^2} \sum_{j,j' = 1}^n \chi(x_{k,j}^\gen = x_{k',j'}^\gen).
\end{equation}

\begin{flushleft}
\textit{$F_{st}$:}\\
\end{flushleft}

For $\phi_1 \ne 1$
\begin{equation}  \label{E:F_st_IA}
F_{st} = \frac{\phi_0 - \phi_1}{1 - \phi_1}.
\end{equation}

\subsection{Coalescent Models}  \label{S:CM}
	
	We model the evolution of a structured population by specifying a coalescent process.  Coalescent processes are Markov jump processes.  We start by defining the state space for these coalescent processes.  We use the notation found in \cite{Limic_2006_Elec_J_Prob}.
	
	Let $\G = \{g_1, g_2, \dots, g_D\}$.  $\G$ represents the demes composing the population.    Let $\F = \bigcup_{k=1}^d \bigcup_{j=1}^n \{x_{k,j}\}$.  $\F$ is the set of all individuals sampled from the population.  Note that $x_{k,j}$ is simply an element of $\F$ serving to represent the $j$th sampled individual from the $k$th sampled deme as oppose to $x_{k,j}^\gen$ which represents genetic data.  Let $\F_k =  \bigcup_{j=1}^n \{x_{k,j}\}$.  $\F_k$ is the set of individuals sampled from the $k$th sample deme.  Let $\P$ be the set of partitions of $\F$.  A partition of $\F$ corresponds to a collection of disjoint sets $ E_1, E_2, \dots, E_m$ such that $\bigcup_{i=1}^m E_i = \F$.  We specify $\pi \in \P$ by $\pi = \{E_1, E_2, \dots, E_m\}$, and refer to the $E_i$ as the \textit{blocks} of $\pi$.   Let  $\PL$ be the set of partitions of $\F$ in which each block is assigned a label from $\G$.  That is,
\begin{equation}
\PL = \{\{(E_1, g_1), (E_2, g_2), \dots, (E_m, g_m)\} : \bigcup_{i=1}^m E_i = \F, g_i \in \G\}
\end{equation}
Intuitively, $g_i$ is the deme occupied by block $E_i$.  For $\pi \in \PL$ we let $|\pi|$ represent the number of blocks forming $\pi$.  We define a coalescent process as a Markov process in which only two type of state jumps are possible.
\begin{enumerate}
\item A labeled block $(E,a)$ may change to $(E, a')$.  This is a migration event.
\item Two blocks $(E_1,a)$ and $(E_2,a)$ may combine to form a single block $(E_1 \cup E_2, a)$.  This is a coalescent event.
\end{enumerate}

	  We let $\Pi(t)$ represent the state of a coalescent process at time $t$.  So $\Pi(t) \in \PL$.   The different coalescent processes are specified through their different transition probabilities.  We first consider three standard coalescent processes: the Kingman coalescent, island model coalescent, and stepping stone model coalescent.
	
\begin{flushleft}
\textit{Kingman Coalescent}:
\end{flushleft}
	
	We denote the Kingman coalescent by $\Pi_\KC(t)$.  In the Kingman coalescent we have $D=1$ and so we can ignore the labels of the blocks.  The jump rates of $\Pi_\KC(t)$ are given by the following rule:
\begin{flushleft}
Two blocks $\{E_i\}$ and $\{E_j\}$ coalesce into $\{E_i \cup E_j\}$ at rate $1$.
\end{flushleft}

\begin{flushleft}
\textit{Island Model Coalescent}:
\end{flushleft}

	We denote the island model coalescent by $\Pi_\IM(t)$.  In this model we set $\G = \{1,2,\dots,D\}$.   The jump rates of $\Pi_\IM(t)$ are given by the following rule:
\begin{enumerate}
\item The labeled block $\{E_i, a_i\}$ migrates to $\{E_i, a_i'\}$ at rate $\frac{m}{D}$.
\item Two labeled blocks $\{E_i, a\}$ and $\{E_j, a\}$ coalesce into $\{E_i \cup E_j, a\}$ at rate $\frac{1}{N}$.
\end{enumerate}
$m$ is the migration rate.  The island model is a completely symmetric model, a migrant is equally likely to migrate to any deme.

\begin{flushleft}
\textit{Stepping Stone Coalescent}:
\end{flushleft}

	We denote the stepping stone model coalescent by $\Pi_\SS(t)$.  In this model we let $\G$ be the lattice in $\mathbb{Z}^2$ specified  by $[0, 1,2,\dots,W-1] \times [0,2,\dots,W-1]$.  To make a connection to the island model case we set $D = W^2$.  We think of $\G$ as a torus.  The neighbor demes of deme $(i,j)$ are $(i+1,j), (i-1,j), (i,j+1), (i,j-1)$ where the arithmetic is modulo $W$.   The jump rates of $\Pi_\SS(t)$ are given by the following rules:
\begin{enumerate}
\item A block $E_i$ migrates from its current deme to a neighboring deme at rate $\frac{m}{4}$.  
\item If two blocks, $E_i$ and $E_j$, occupy the same deme then they coalesce at rate $\frac{1}{N}$.  
\end{enumerate}

	In all the models we consider, genetic diversity is created by mutations.  To model mutation, we assume that blocks experience mutations at rate $\mu$. At $t=0$, we set $x_{k,j}^\gen(i) = 0$ for all $k,j,i$.  We let $e(t)$ be the mutation counter.  That is, $e(0) = 0$ and every time a mutation occurs $e(t)$ is incremented by $1$.  When a block, say $E$, mutates we set $x_{k,j}^\gen(e(t)) = 1$ for every $x_{k,j} \in E$.  Often, in the case of the Kingman coalescent we will make the mutation rate explicit by writing $\Pi_\KC(t,\mu)$.  For the island and stepping stone model coalescents we define $\theta = \mu N D$.  
	
	While diversity measures are defined as functions on the $x_{k,j}^\gen$, the value of each $x_{k,j}^\gen$ is determined by the underlying coalescent.  For this reason we write $G(\Pi(t))$ to mean $G$ under the coalescent process $\Pi(t)$.
	
\subsection{The G/KC coalescent}  \label{S:GKC}

	In \cite{Wakeley_1999_Genetics}, Wakeley pointed out that the dynamics of $\Pi_\IM(t)$ can be decomposed into two phases: a scattering phase and a collecting phase.  The scattering phase describes the initial phase of $\Pi_\IM(t)$ in which blocks migrate away their start demes until every block occupies a separate deme.   Then, in the collecting phase,  blocks that occupy separate demes migrate to common demes and coalesce until a single block remains.  As Wakeley pointed out, the collecting phase is well modeled by the Kingman coalescent.
	
	We distill three key components of the scattering-collecting decomposition that can be applied in a more general setting than the island model.
	
\begin{enumerate}
\item During the scattering phase, no two blocks that contain individuals from separate sampled demes coalesce.
\item The scattering phase occurs on a much faster time scale then the collecting phase.
\item During the collecting phase, the coalescent is well described by the Kingman coalescent
\end{enumerate}

	We introduce a coalescent process that is a limiting version of these three requirements.  We refer to this coalescent as the G/KC coalescent and denote it $\Pi_\GKC(t)$.  Like $\Pi_\KC(t)$, the blocks of $\Pi_\GKC(t)$ are not labeled.  To define $\Pi_\GKC$ we specify a random partitioning of each $\F_k$.   More precisely, we assume that $\F_k$ is partitioned into $B_k$ blocks, $E_{k,1}, E_{k,2}, \dots, E_{k,B_k}$. We set $b_{k,j} = \frac{|E_{k,j}|}{n}$ which implies $b_{k,1} + b_{k,2} + \dots + b_{k,B_k} = 1$.  For $k=1,2,\dots,d$, the tuples $(B_k, b_{k,1}, \dots, b_{k,B_k})$ are i.i.d.  Since diversity measures are symmetric in the individuals forming each $\F_k$, we need only specify $|E_{k,j}|$.  
	
	The random partitioning is then used to define the initial condition of the $G/KC$ coalescent.
\begin{equation}
\Pi_\GKC(0) = \cup_{k=1}^d \cup_{j=1}^{B_k} \{E_{k,j}\}.
\end{equation}
The dynamics of the G/KC coalescent are given by a Kingman coalescent with mutation rate $r$.  That is, for some $r > 0$
\begin{equation}
\Pi_\GKC(t) = \Pi_\KC(t;r)
\end{equation}
$\Pi_\GKC$ is simply the Kingman coalescent run at mutation rate $r$ with a random initial partitioning of the $\F_k$.  The G/KC coalescent is specified by $r$ and the distribution of $(B_k, b_{k,1}, \dots, b_{k,B_k})$.

	The G/KC coalescent is a limiting version of Wakeley's scattering-collecting decomposition.  The scattering phase, which occurs on a fast time scale for $\Pi_\IM(t)$, is instantaneous in $\Pi_\GKC(t)$ and is completely general in its distribution (hence the G in G/KC).  The collecting phase, which occurs on a slower time scale, is described by the Kingman coalescent (hence the KC in G/KC).
	
	To each G/KC coalescent we associate scattering probabilities.  Let $I,j_1, \dots, j_I$ be positive integers and set $J=j_1 + j_2 + \dots + j_I$.  Suppose we select $J$ individuals from $\F_k$.  Then $\SP(j_1,j_2,\dots,j_I)$ is the probability that the $J$ individuals are partitioned into $I$ sets of size $j_1, j_2, \dots, j_I$ by the blocks $E_{k,1}, E_{k,2}, E_{k,B_k}$.  We refer to $\SP(j_1,j_2,\dots,j_I)$ as a scattering probability.

\section{Results}  \label{S:results}
\setcounter{equation}{0}
\setcounter{lemma}{0}

	We consider diversity measures in the large population, large sample (LPLS) limit which we write as $\LPLSL$ and define as follows.  In the LPLS limit we take $N, D, n, d \to \infty$.  The limit requires some further assumptions depending on the coalescent process we are considering. When we consider the island model, we set $\Gamma = Nm$ and assume that $\Gamma, \theta$ are held fixed while $\frac{(nd)^2}{\sqrt{D}} \to 0$, $\frac{\log^2(n)}{d} \to 0$.   In the case of the stepping stone model we follow \cite{Cox_2002_Ann_Appl_Prob} by setting $\alpha = \frac{2\pi Nm}{\log W}$.   We then fix $\alpha, \theta$  while $\frac{(nd)^2 \log \log W}{\sqrt{\log W}} \to 0$. We also require that sample demes are separated by a distance of at least $\Dsample = \frac{W}{\sqrt{\log W}}$.   In considering G/KC coalescents, we fix $r$, assume the tuples $(B_k, b_{k,1}, \dots, b_{k,B_k})$ converge in distribution, and take $\frac{E[B_1^2]}{d} \to 0$.   Since $(B_k, b_{k,1}, \dots, b_{k,B_k})$ converges, the limit of $\SP$ exists and we set $\SP \to \SPL$.  Whenever we refer to a limit, we are considering the LPLS limit unless we specify otherwise.  
	
	Our first two results demonstrates that the analysis of diversity measures under the island  or stepping stone model coalescents can be reduced to the analysis of diversity measures for G/KC coalescents.   Define
\begin{equation}
\Upsilon_j = \beta_j \prod_{i=1}^{j-1} (1 - \beta_{i})
\end{equation}
and where the $\beta_j$ are i.i.d as Beta$[1, 2\Gamma]$.  Then we have the following result.
	
\begin{theorem}  [Island Model Convergence] \label{T:IM_convergence}
Let G be a diversity measure.  Then,
\begin{equation}  \label{E:IM_G_convergence}
\LPLSL G(\Pi_\IM(t)) = \LPLSL G(\Pi_\GKC(t))
\end{equation}
where $r = \theta \frac{1 + 2\Gamma}{2\Gamma}$, $B_k \to \infty$, and for fixed $J$,
\begin{equation}
(b_{k,1}, b_{k,2}, \dots, b_{k,J}) \to 
	 (\Upsilon_1, \Upsilon_2, \dots, \Upsilon_J)
\end{equation}
\end{theorem}

\begin{theorem}  [Stepping Stone Model Convergence] \label{T:SS_convergence}
\begin{equation}  \label{E:SS_G_convergence}
\LPLSL G(\Pi_\SS(t)) = \LPLSL G(\Pi_\GKC(t))
\end{equation}
where $r = \theta \frac{1 + \alpha}{\alpha}$ and $\LPLSL (B_k, b_{k,1}, \dots, b_{k,B_k})$ is distributed as the blocks of $\Pi_\KC^{(\infty)}(\log(\frac{1 + \alpha}{\alpha}))$.
\end{theorem}	

	The next result characterizes the distribution of $F_{st}$ under a G/KC coalescent.  We split into two cases.  First, we consider the case of  $r \to \infty$, which we refer to as the strong mutation case.

\begin{theorem} [Strong Mutation Case] \label{T:Fst_strong}
\begin{equation}
\lim_{r \to \infty} \LPLSL F_{st}(\Pi_\GKC(t)) = \SPL(2).
\end{equation}
\end{theorem}

	Taking $r \to 0$ corresponds to the assumption of weak mutation.  In computing $ F_{st}$ under weak mutation we may assume that exactly one mutation occurs in the G/KC coalescent.  We assume that the mutation occurs when $|\Pi_\GKC(t)| = L$.  Define $\lambda = \LPLSL \frac{L}{d}$ and $\kappa = \LPLSL \frac{L}{E[B_1]d}$.  The following theorem shows that when $\lambda = 0$, the distribution of $F_{st}$ in the weak mutation case is the same as that in the strong mutation case.
\begin{theorem} [Weak Mutation Case] \label{T:bottom_tree}
If $\lambda = 0$ then,
\begin{equation}
\LPLSL F_{st}(\Pi_\GKC(t)) = \SPL(2).
\end{equation}
\end{theorem}

	If $\lambda  \ne 0$, then the following results show that $F_{st}$ has a non-degenerate distribution.  

\begin{theorem}  [Weak Mutation Case] \label{T:middle_tree}
Assume $\lambda > 0$ and $\kappa = 0$.  
\begin{equation}
\LPLSL F_{st}(\Pi_\GKC(t)) = \frac{\sum_{k=1}^Q X_k^2}{\sum_{k=1}^Q X_k}
\end{equation}
where $X_k$ are i.i.d. versions of the random variable $X$ which is defined by the following moment relations
\begin{equation}
E[X^k] = \SPL(k)
\end{equation}
and $Q$ is Poisson distributed with rate $\frac{V}{\lambda}$, where $V$ is exponentially distributed with mean $1$.
\end{theorem}

\begin{theorem}  [Weak Mutation Case]  \label{T:top_tree}
Assume $\lambda > 0$ and $0 < \kappa < 1$. Let $G(\kappa)$ be a geometric random variable with success probability $\kappa$.  If $\LPLSL E[B_k] < \infty$ then
\begin{equation}
\LPLSL F_{st}(\Pi_\GKC(t)) = \frac{\sum_{k=1}^{G(\kappa)+1} W_k^2}{\sum_{k=1}^{G(\kappa)+1} W_k}
\end{equation}
where $W_k$ are i.i.d. versions of the random variable $W$ which is defined by the following moment relations
\begin{equation}
E[W^k] = \frac{\SPL(k)}{\LPLSL E[B_k]}
\end{equation}
If $\LPLSL E[B_k] = \infty$ then
\begin{equation}
\LPLSL F_{st} = 0.
\end{equation}
\end{theorem}

	Theorems \ref{T:IM_convergence} and \ref{T:SS_convergence} are proved in section \ref{S:convergence}.  Theorem \ref{T:Fst_strong} is proved in section \ref{S:strong}.  Theorems \ref{T:bottom_tree}-\ref{T:top_tree} are proved in section \ref{S:weak}.


\section{Applications}  \label{S:applications}
\setcounter{equation}{0}
\setcounter{lemma}{0}

	We now apply the results stated in section \ref{S:results}.  In section \ref{S:Fst} we examine the distribution of $F_{st}$ under a single locus, infinite allele model.  In section \ref{S:Fst_seq} we examine $F_{st}$ under a multiple locus, biallelic model and under an infinite sites model.  Finally in section \ref{S:Homo}, we consider homozygosity measures.  
	
\subsection{$F_{st}$}  \label{S:Fst}

	$F_{st}$ as defined in (\ref{E:F_st_IA}) corresponds to a single locus, infinite allele model.  In such a setting, the distribution of $F_{st}$ has been a subject of research for some time.  The relation $F_{st} = \frac{1}{1 + 2Nm}$ was originally proposed by Sewall Wright \cite{Wright_1942_Genetics}. A quantity related to $F_{st}$, which we label $F_{st}^*$, is defined by
\begin{equation}  \label{E:Fst_star}
F_{st}^* = \frac{E[\phi_0 - \phi_1]}{E[1 - \phi_1]}.
\end{equation}
 	In \cite{Slatkin_1991_Genet_Res_Camb, Wilkinson_Herbots_1998_J_Math_Bio} the authors derive the value of $F_{st}^*$ under the island model, while in \cite{Leviyang_2008, Rottenstreich_2007a_Theo_Pop_Bio}, the authors derive the distribution of $F_{st}$ for the island model in the strong and weak mutation cases. In \cite{Cox_2002_Ann_Appl_Prob} the authors derive the value of $F_{st}^*$ for the stepping stone model.  We note that while $F_{st}^*$ is a quantity of theoretical interest, $F_{st}$ is more relevant in application.  $F_{st}$ is a random variable, while $F_{st}^*$ is deterministic.  In this paper we consider $F_{st}$.  The previous results leave two fundamental questions unanswered.  
\begin{enumerate}
\item How is the distribution of $F_{st}$ affected by changes in the structured population model?
\item What is the distribution of $F_{st}$ for the stepping stone model?
\end{enumerate}

	The first question is answered by Theorems \ref{T:Fst_strong}-\ref{T:top_tree} for populations that converge to G/KC coalescents.  In the strong mutation case, $F_{st}$ will converge to a deterministic limit, while in the weak mutation case the distribution of $F_{st}$ can be computed and will depend on where in the coalescent the mutation occurs.  
	
	Now we turn to the second question and consider $F_{st}$ for the stepping stone model.  For completeness, we will also state the corresponding results for the island model.  To compute LPLS limits of $F_{st}$ we need to compute $\SPL(k)$ for $k \ge 2$.    For the island mode, $\SPL(k)$ is the probability that $k$ individuals in a given deme all coalesce before a migration occurs.  Simple Kingman coalescent arguments, see \cite{Durrett_Book_Probability_Models}, give
\begin{equation}  \label{E:SPL_IM}
\SPL(k) = \prod_{j=1}^{k-1} \frac{j}{j + 2\Gamma}.
\end{equation}
Note that $\SPL(2) = \frac{1}{1 + 2\Gamma}$.  For the stepping stone model, $\SPL(k) = P(|\Pi_\KC^{(k)}(\log(\frac{1 + \alpha}{\alpha})) = 1)$.  By equation 5.2 in \cite{Tavare_1984_TPB}, 
\begin{equation}  \label{E:SPL_SS}
\SPL(k) = f_k(\log(\frac{1+\alpha}{\alpha})),
\end{equation}
where  
\begin{equation}
f_k(t) = 1 + \sum_{h=2}^k  \exp[-(\frac{h(h-1)}{2})t] (-1)^{h-1} (2h-1) \binom{h}{2} \frac{\binom{k}{h}}{\binom{k+h-1}{h}}.
\end{equation}
Note that this gives $\SPL(2) = \frac{1}{1 + \alpha}$.  Using (\ref{E:SPL_IM}) and (\ref{E:SPL_SS}) and Theorem \ref{T:Fst_strong} we can compute the LPLS limits of $F_{st}$.  For strong mutation we have the following result.
\begin{proposition}  \label{P:IM}
Let $\theta \to \infty$. Then for the island model
\begin{equation}
F_{st} \to \frac{1}{1 + 2\Gamma}
\end{equation}
while for the stepping stone model
\begin{equation}
F_{st} \to \frac{1}{1 + \alpha}.
\end{equation}
\end{proposition}
The exact same result holds in the case of weak mutation when $\lambda = 0$.  For $\lambda > 0$, we can numerically compute the distribution of $F_{st}$.  For instance, consider the case $\lambda = 2$.  In this case $\kappa = \LPLSL \frac{2}{E[B_k]}$.  For the island model $B_k \to \infty$, so $\kappa = 0$.  For the stepping stone model, $E[B_k]$ is finite and can be numerically computed using known formulas \cite{Tavare_1984_TPB} (we find $\kappa \approx .388$).  Using Theorem \ref{T:middle_tree} for the island model and Theorem \ref{T:top_tree} for the stepping stone model we can numerically compute the distribution of $F_{st}$.  The result is given in  figure \ref{Fi:Fst_figure} in the case $\Gamma = 1$ for the island model and $\alpha = 2$ for the stepping stone model.  In this case the mean of $F_{st}$ for the island and stepping stone models is approximately $.2$ and $.1$ respectively.  

\begin{figure} [htp]  
\begin{center}
\includegraphics[width=1\textwidth]{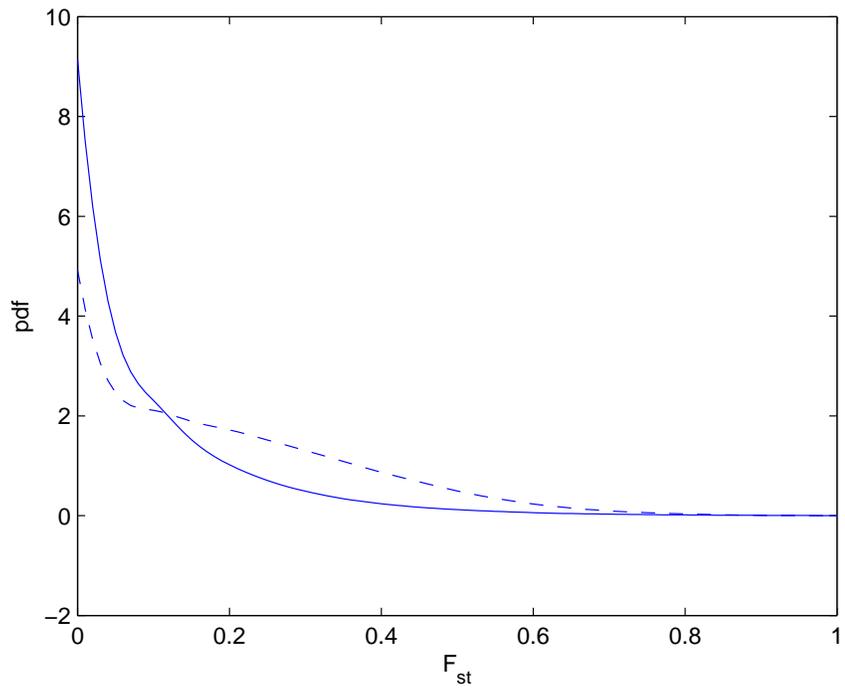}
\caption{pdf of $\LPLSL F_{st}$ for island model (dashed line) with $\Gamma = 1$ and stepping stone model (unbroken line) with $\alpha = 2$. In both cases $\lambda = 2$}
\label{Fi:Fst_figure}
\end{center}
\end{figure}

\subsection{Generalizations of $F_{st}$}  \label{S:Fst_seq}

	Today, genetic data  rarely fits the single locus, infinite alleles assumption of the previous section.  We examine two generalizations of $F_{st}$.   In \cite{Weir_1984_Evolution}, Weir and Cockerham generalized $F_{st}$ to biallelic, multiple loci data.  To model such data we let $x_{k,j}^\gen(i)$ represents the allelic state of locus $i$ for the given individual.  In \cite{Hudson_1992_Mol_Bio_Evol, Lynch_1990_Mol_Bio_Evol, Nei_1982_Human_Genetics} the authors consider $F_{st}$ generalized to sequence data.  In this setting, we let $x_{k,j}^\gen$ represent a string of $0$s and $1$s.
	
	We start by considering biallelic, multiple loci data.  We assume $l$ loci and a single mutation on $\Pi_\GKC(t)$ for each locus.  Define for $i=1,\dots,l$
\begin{equation}   \label{E:formula_phi_0_k_i}
\phi_{0,k}(i) = \frac{1}{d} \sum_{k=1}^d \frac{1}{n^2} \sum_{j,j' = 1}^n \chi(x_{k,j}^\gen(i) = x_{k,j'}^\gen(i)),
\end{equation}

\begin{equation}   \label{E:formula_phi_0_i}
\phi_0(i) = \frac{1}{d} \sum_{k=1}^d \phi_{0,k}(i),
\end{equation}

\begin{equation}  \label{E:formula_phi_1_i}
\phi_1(i) = \frac{1}{d^2} \sum_{k,k'=1}^d \frac{1}{n^2} \sum_{j,j' = 1}^n \chi(x_{k,j}^\gen(i) = x_{k',j'}^\gen(i)).
\end{equation}
$\phi_0(i), \phi_1(i)$ are homozygosity measures for locus $i$, and we can use these measures to form an $F_{st}$ value for each locus.  
\begin{equation}
F_{st,i} = \frac{\phi_{0}(i) - \phi_{1}(i)}{1 - \phi_{1}(i)}.
\end{equation}
A key question considered by Weir and Cockerham is how to combine the $\phi_0(i), \phi_1(i)$ values in order to produce a statistic with small variance.  In a widely cited paper, \cite{Weir_1984_Evolution}, Weir and Cockerham suggested using $F_{st}^\IS$, where for $\sum_{i=0}^\infty (1 - \phi_1(i)) \ne 0$
\begin{equation}  \label{E:Fst_IS}
F_{st}^\IS = \frac{\sum_{i=0}^l \phi_0(i) - \phi_1(i)}{\sum_{i'=0}^l (1 - \phi_1(i'))}.
\end{equation}
Alternatively, one might form a statistic by simply averaging the $F_{st,i}$.  That is,
\begin{equation}
F_{st}^\AV = \frac{1}{l} \sum_{i=1}^l F_{st,i}
\end{equation}

	Our analysis of $F_{st}$ allows us to prove the following result.
\begin{proposition}  \label{P:Fst_IS_limit}
Fix $l$.  
\begin{equation}
F_{st}^\IS(\Pi_{\GKC}(t)) \to \SPL(2)
\end{equation}
\begin{equation}
F_{st}^\AV(\Pi_{\GKC}(t)) \to \SPL(2)
\end{equation}
\end{proposition}

 To see Proposition \ref{P:Fst_IS_limit} first note that we may assume that $\Pi_{\GKC}(t)$ has exactly $l$ mutations.  Let the levels of these mutations be $L_1, L_2, \dots, L_l$ where each $L_i$ is i.i.d.  Theorem \ref{T:Fst_strong} shows that if $\frac{L_i}{d} \to 0$ for all $i$, then each $F_{st,i} \to \SPL(2)$ and the result will follow.   Let $T_k = \inf \{t : |\Pi_\GKC(t)| = k\}$.   Since mutations are distributed as a Poisson process we have
\begin{equation}
P(L_i) = \frac{L_i (T_i - T_{i-1})}{\sum_{j=2}^{|\Pi_\GKC(0)|} j (T_j - T_{j-1})}
\end{equation}
Well known results for the Kingman coalescent, see for example section 1.3.1 of \cite{Delmas_2008_Ann_Appl_Prob}, give $\sum_{j=2}^{|\Pi_\GKC(0)|} j (T_j - T_{j-1}) = O(\log(|\Pi_\GKC(0)|))$ while $L_i (T_i - T_{i-1}) = O(\frac{1}{L_i})$. Using these results and noting $d < |\Pi_\GKC(0)| < nd$ gives for fixed $\delta > 0$
\begin{equation}
P(\frac{L_i}{d} > \delta) \le O(\frac{|\log \delta|}{\log d}).
\end{equation}
This shows that for fixed $l$, we have $\frac{L_i}{d} \to 0$.  In fact as long as $l \ll \log d$ the result holds.

	In \cite{Weir_1984_Evolution}, Weir and Cockerham showed through numerical experiments that for finite samples $F_{st}^\IS$ has lower variance then $F_{st}^\AV$.  To explain this, we note that $F_{st,i}$ will have high variance if $\frac{L_i}{d}$ is $O(1)$.  We will also eventually show, see Lemmas \ref{L:mean_p} and \ref{L:var_q}, that the means of $\phi_0(i), \phi_1(i)$ are $O(\frac{1}{L})$ while variances are  $O(\frac{1}{L^2})$.  Now suppose that $\frac{L_1}{d} = 1$  while for $i \ne 1$, $L_i = O(1)$.  In this case, with $F_{st}^\AV$ in mind, we have the following facts.  
\begin{itemize}
\item $V[F_{st,1}] = 0(1)$.
\item For $i \ne 1$, $F_{st,i} \approx \SPL(2)$ and $V[F_{st,i}] = o(1)$.
\end{itemize}
	These two facts give $V[F_{st}^\AV] = O(\frac{1}{l^2})$.  For $F_{st}^\IS$, the following facts are relevant.
\begin{itemize}
\item $E[\phi_0(1)] = O(\frac{1}{d})$, $E[\phi_1(1)] = O(\frac{1}{d})$.
\item For $i \ne 1$, $V[\phi_0(i)] = O(1)$ and $V[\phi_1(i)] = O(1)$.
\item For $i \ne 1$, $\frac{\phi_0(i) - \phi_1(i)}{1 - \phi_1(i)} \approx \SPL(2)$ and $V[\frac{\phi_0(i) - \phi_1(i)}{1 - \phi_1(i)}] = o(1)$.
\end{itemize}
	These three facts give $V[F_{st}^\IS] = O(\frac{1}{ld})$.  For $d \gg l$ we see that $F_{st}^\IS$ has lower variance than $F_{st}^\AV$.
	
	Now we consider $F_{st}$ for sequence data.  Various formulas exist for such a generalization, see \cite{Hudson_1992_Mol_Bio_Evol} for a summary, but up to small variations all are given by the formula for $F_{st}^\IS$ given in (\ref{E:Fst_IS}) with $l = \infty$.  This means, if we assume a fixed number of mutations,  that  our analysis from Proposition \ref{P:Fst_IS_limit} holds and we have $F_{st}^\IS(\Pi_\GKC(t)) \to \SPL(2)$ for sequence data.

\subsection{Homozygosity Measures}  \label{S:Homo}

	Homozygosity measures are commonly used to quantify genetic diversity.  Previous work on homozygosity measures for subdivided populations has focused on computing means for $\phi_{0,k}$ and $\phi_{1,k}$, e.g. \cite{Maruyama_1971_Genetics, Nagylaki_1974_PNAS}.  In this section we derive the distribution of  $\phi_{0,k}$ under the infinite alleles model and the assumption of strong mutation.  By the definition of the G/KC coalescent, at $t=0$ the $n$ individuals from sampled deme $k$ are split into $B_k$ blocks of relative sizes $b_{k,1}, b_{k,2}, \dots, b_{k,B_k}$.  If mutation is sufficiently strong, $r \gg 1$, each of these blocks will experience a mutation prior to a coalescent event.  In such a case, each of the $B_k$ blocks will have a different allelic state.  This allows us to compute the distribution of $\phi_{0,k}$.
\begin{equation}
\phi_{0,k} = \sum_{j=1}^{B_k} b_{k,j}^2.
\end{equation}

	For the case of the island model,  Theorem \ref{T:IM_convergence} gives
\begin{equation}
\phi_{0,k} \to \sum_{j=1}^\infty \Upsilon_j^2.
\end{equation}
For the case of the stepping stone model, let $V_{j,i}$ be exponential random variables with mean $1$ that are independent over $i,j$ for $i=1,2,\dots$ and $j=1,\dots,i$. Then, one can show (see \cite{Durrett_Book_Probability_Models}) that
\begin{equation}
b_{k,j}\bigg|_{B_k} \to \frac{V_{B_k,j}}{V_{B_k,1} + V_{B_k,2} + \dots + V_{B_k,B_k}}.
\end{equation}
Theorem \ref{T:IM_convergence} then gives
\begin{equation}
\phi_{0,k} \to \sum_{i=1}^\infty h_i(\log(\frac{1+\alpha}{\alpha})) 
	\frac{V_{i,1}^2 + \dots V_{i,i}^2}{(V_{i,1} + V_{i,2} + \dots + V_{i,i})^2},
\end{equation}
where $h_i$, by equation 5.2 in \cite{Tavare_1984_TPB}, is defined as
\begin{equation}
h_i(t) = 
\bigg\{
\begin{array}{cc}
\sum_{k=i}^\infty \exp[-(\frac{k(k-1)}{2})t] (\frac{2k-1}{k-1}) (-1)^{k-i}
		\binom{i+k-2}{i} \binom{k-1}{i-1} & \text{ if }i \ne 1 \\
1 + \sum_{k=2}^\infty \exp[-(\frac{k(k-1)}{2})t] (\frac{2k-1}{k-1}) (-1)^{k-i}
		\binom{i+k-2}{i} \binom{k-1}{i-1}& \text{ if } i = 1.
\end{array}
\end{equation}

	Under strong mutation, the LPLS limit distributions of $\phi_{0,k}$ for the island model and stepping stone model are given in figure \ref{Fi:phi_0_figure}.  As in section \ref{S:Fst}, we take $\Gamma = 1$ for the island model case, and $\alpha = 2$ in the stepping stone case.  This gives, for both cases, $E[\phi_{0,k}] = \frac{1}{3}$.  We note the similarity in the distribution of $\phi_{0,k}$ under the two models.  Currently, there are many statistical tests for population subdivision, but we are not aware of any statistical test that addresses the type of subdivision.  The similarity in homozygosity measures for the island model and stepping stone model suggests that any such test should not involve homozygosity measures.
	
\begin{figure} [htp] 
\begin{center}
\includegraphics[width=1\textwidth]{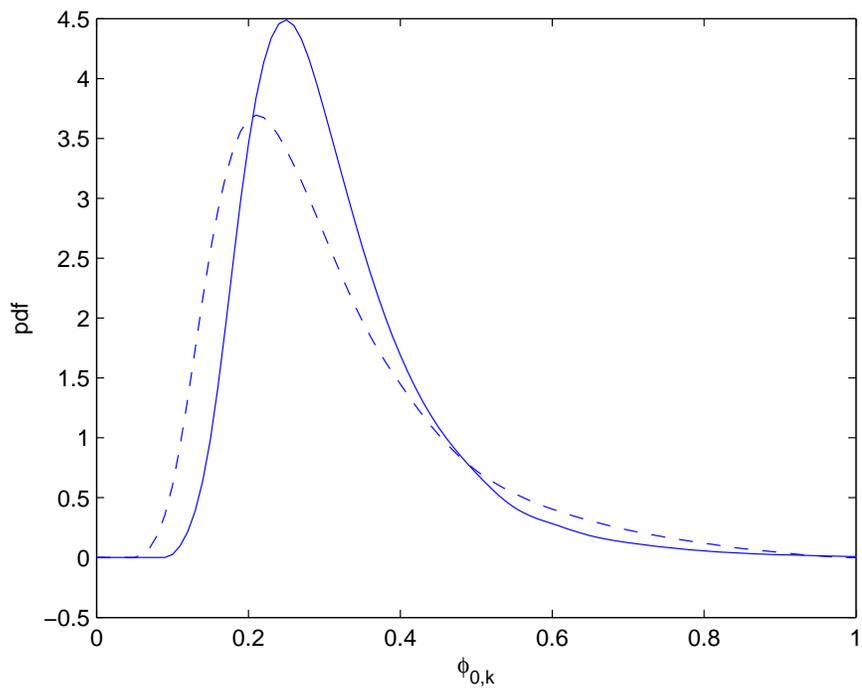}
\caption{pdf of $\phi_{0,k}$ for island model (dashed line) with $\Gamma = 1$ and stepping stone model (unbroken line) with $\alpha = 2$}
\label{Fi:phi_0_figure}
\end{center}
\end{figure}

\section{Convergence to the G/KC coalescent}  \label{S:convergence}
\setcounter{equation}{0}
\setcounter{lemma}{0}
	
	In this section we prove Theorems \ref{T:IM_convergence} and \ref{T:SS_convergence}.  To do this, we define a time $\Td$ and show that the following conditions hold.
\begin{enumerate}
\item (Independence Condition) The probability that individuals from separate sampled demes coalesce before $\Td$ goes to zero.
\item (Short Scattering Phase Condition)  The probability of a mutation before $\Td$ goes to zero.
\item (KC Condition) After time $\Td$, the coalescent converges to a Kingman coalescent
\end{enumerate}
After demonstrating these three condition, we determine the distribution of $B_k, b_{k,1}, b_{k,2}, \dots, b_{k,B_k}$ formed by $\Pi_\GKC(\Td)$.  Lastly, we show that for both Theorem \ref{T:IM_convergence} and \ref{T:SS_convergence} the condition $\frac{E[B_1^2]}{d} \to 0$ holds.

	To demonstrate the KC condition, we introduce the following notation.  For a general coalescent process $\Pi(t)$, let $E_1, E_2, \dots, E_k$ be the blocks forming $\Pi(T_k)$.  Recall $T_k = \inf\{t : |\Pi(t)| = k\}$.   Define $N_j(k \to k-1)$ as the number of mutations that block $E_j$ experiences during time $[T_k, T_{k-1})$.  Let $U_1(k)$, $U_2(k)$ be the indices of the two blocks that coalesce at time $T_{k-1}$.  If we specify some unique way of ordering the blocks $E_i$ (say by ordering $E_i$ based on some lexographic ordering of the $x_{k,j}$) then any diversity measure $G$ will be a function of $U_1(k), U_2(k)$,$N_j(k \to k-1)$, and $\Pi(0)$.
	  
	We will use the following Lemma to prove the KC condition. 
\begin{lemma}  \label{L:Fst_limit}
Let $\Pi(t)$ be a coalescent process with $|\Pi(0)| = M$ and let $\Pi_\KC(t,r)$ be a Kingman coalescent with $\Pi_\KC(0)$  equal to $\Pi(0)$ with the labels of the blocks removed.  Let $G$ be a diversity measure.  If
\begin{equation}  \label{E:N_j_general}
\sum_{k=2}^M \sum_{j=1}^k |E[N_j(k \to k-1)] - \frac{r}{\binom{k}{2}}| \to 0
\end{equation}
and
\begin{equation}  \label{E:U_general}
\sum_{k=2}^M \bigg|1 - k(k-1) \inf_{j,j'=1,\dots,k; j\ne j'} P(U_1(k) = j, U_2(k) = j')\bigg| \to 0
\end{equation}
then
\begin{equation}
\LPLSL G(\Pi(t)) = \LPLSL G(\Pi_\KC(t,r)).
\end{equation}
\end{lemma}

\begin{proof}
Since mutation events are Poisson processes, $N_j(k \to k-1)$ has a Poisson distribution.  Let $\tilde{N}_j(k \to k-1)$ be the $N_j$ associated with $\Pi_\KC(t,r)$.  Then $\tilde{N}_j(k \to k-1)$ has Poisson distribution with mean $\frac{r}{\binom{k}{2}}$.  We couple mutation events on $\Pi(t)$ and $\Pi_\KC(rt)$ for their respective intervals $[T_k, T_{k-1})$ as follows.  Match the blocks $k$ blocks in $\Pi(t)$ with the $k$ blocks in $\Pi_\KC(t)$ in some arbitrary way.  Apply mutations to each block according to a Poisson distribution with mean $\frac{r}{\binom{k}{2}}$.  Now add more mutations to each block in $\Pi(t)$ according to a Poisson distribution with mean $E[N_j(k \to k-1)] - \frac{r}{\binom{k}{2}}$ (if the quantity is negative, remove mutations).  If we add (or remove) mutations in this second step we say that a decoupling event has taken place.  By (\ref{E:N_j_general}) the probability of a decoupling event over all $k$ goes to zero.  So we have a coupling between the mutations on $\Pi(t)$ and $\Pi_\KC(t)$.  

	Now we establish a coupling for $U_1,U_2$.  Let $\tilde{U}_1, \tilde{U}_2$ be the $U_i$ corresponding to $\Pi_\KC(t)$.  Then $P(\tilde{U}_1(k) = j, \tilde{U}_2(k) = j') = \frac{1}{k(k-1)}$.  Set $a = \inf_{j \ne j'} P(U_1(k) = j, U_2(k) = j')$.  We now partition $[0,1]$ into $k(k-1) + 1$ intervals.  $k(k-1)$ of these intervals are of size $a$ and each of these intervals corresponds to a specific $j,j'$ combination.  We couple $U_i$ and $\tilde{U}_i$ as follows.  We select a number uniformly on $[0,1]$.  If the number lands in one of the $k(k-1)$ intervals corresponding to some $j,j'$ pair then we coalesce the same blocks in $\Pi(t)$ as coalesce in $\Pi_\KC(rt)$.  Otherwise, if the number falls in the interval that does not correspond to a $j,j'$ pair, we say a decoupling has occurred and we coalesce each process separately.  (\ref{E:U_general}) shows that over all $k$ the probability of a decoupling goes to zero.  So we have a coupling between the blocks that coalesce in $\Pi(t)$ and those that coalesce in $\Pi_\KC(rt)$.

	The result now follows from the observation that $G$ is bounded and depends only on the number of mutations in each block and the order in which the blocks coalesce.

\end{proof}

	Before proceeding we set some notation.  For any coalescent $\Pi(t)$ (that is $\Pi_\KC$, $\Pi_\IM$, $\Pi_\SS$, $\Pi_\GKC$)  we let $\Pi_k(t)$ for $k=1,\dots,d$ represent $\Pi(t)$ with the blocks intersected against $\F_k$.  That is, if 
\begin{equation}
\Pi(t) =  \{(E_1, a_{1}), (E_2, a_{2}), \dots, (E_m, a_{m})\},
\end{equation}
then
\begin{equation}
\Pi_k(t) \{(E_1 \cap \F_k, a_{1}), (E_2 \cap \F_k, a_{2}), \dots, (E_m \cap \F_k, a_{m})\}
\end{equation}
For $\Pi_\IM(t)$ and $\Pi_\SS(t)$, unless specified otherwise, we take $\Pi(0) = \bigcup_{k=1}^d \bigcup_{j=1}^n \{(x_{k,j}, \D(k))\}$, where $\D(k)$ is the deme label $g \in \G$ corresponding to the $k$th sampled deme.  Since the deme labels in $\Pi_\KC(t)$ may be ignored, $\Pi_\KC(0)$ is specified by $|\Pi_\KC(0)|$.  We write $\Pi_\KC^{(k)}(t)$ for $\Pi_\KC(t)$ with $|\Pi_\KC(0) = k|$.  
 
 We use $\Theta$, and $I$ to represent various probability events and integrals respectively.    Within a given proof, $\Theta$ and $I$ are consistently used, but their definition varies between proofs.  We use $C$ as an arbitrary constant that may change from line to line.

\subsection{Island Model and G/KC}  \label{S:IMGKC}

	In this section we prove Theorem \ref{T:IM_convergence}.  Set $\Td = N \sqrt{D}$.

\begin{lemma} [Independence Condition]  \label{L:IM_IC}
Let $\Theta$ be the event in which two blocks from separate sampled demes coalesce before time $\Td$.  Then,
\begin{equation}
P(\Theta) = O(\frac{(nd)^2}{\sqrt{D}})
\end{equation}
\end{lemma}

\begin{proof}
A block migrates to a deme that is occupied by another block at a rate bounded by $\frac{m}{D}$.  So the probability of a block entering a deme occupied by another block before time $\Td$ is bounded by 
\begin{equation}
\int_{0}^{\Td} dt \exp[-\frac{m}{D} t] \frac{m}{D} = 1 - \exp[-\frac{\Td m}{D}] = \frac{1}{\sqrt{D}}.
\end{equation}
Summing this probability over all possible pairs gives the result.

\end{proof}

\begin{lemma} [Short Scattering Phase Condition] \label{L:IM_SSPC}
\begin{equation}
E[\text{number of mutation before } \Td] = O(\frac{nd}{\sqrt{D}}).
\end{equation}
\end{lemma}

\begin{proof}
There are at most $nd$ blocks in the time interval $[0, \Td]$.  Then,
\begin{equation}
E[\text{number of mutations before } \Td] \le \mu (nd) \Td = O(\frac{nd}{\sqrt{D}}) \to 0.
\end{equation}

\end{proof}

	Before proving the KC condition, we show that each block of $\Pi_\IM(\Td)$ occupies a separate deme.  We refer to  $\pi \in \PL$ as a scattered state if each block occupies a separate deme.  We refer to $\pi$ as a semi-scattered state if two blocks share the same deme while all other blocks are in separate demes.

\begin{lemma}  \label{L:IM_scatter}
\begin{equation}
P(\Pi_\IM(\Td) \text{ is a scattered state}) \to 1.
\end{equation}
\end{lemma}

\begin{proof}
We demonstrate that the following two facts hold in the LPLS limit.
\begin{itemize}
\item every block experiences at least one migration
\item during $[0,\Td]$, blocks migrate to demes that are unoccupied by other blocks.
\end{itemize}

	To see the first fact recall that blocks migrate at rate $m$.  So the probability of a block not migrating away from its sample deme by time $\Td$ is $O(\exp[-\sqrt{D}])$.  To see the second fact we recall that blocks migrate to a deme occupied by another block at a rate $\frac{m}{D}$.  So the probability of migrating to an occupied deme is $O(\frac{1}{\sqrt{D}})$.  Summing these probabilities over all possible blocks shows that at time $\Td$ every block is in a separate deme with probability $O(\frac{(nd)^2}{\sqrt{D}})$.  Taking the LPLS limit finishes the proof.

\end{proof}

\begin{lemma} [KC condition]   \label{L:IM_KCC}
\begin{equation}  \label{E:N_j_IM}
\sum_{k=2}^{|\Pi_\IM(\Td)|} \sum_{j=1}^k \bigg|E[N_j(k \to k-1)] - \frac{\theta}{\binom{k}{2}} \frac{1 + 2\Gamma}{2\Gamma}\bigg| \to 0. 
\end{equation}
For $j, j'=1,2,\dots,k$ with $j \ne j'$
\begin{equation}  \label{E:U_IM}
\sum_{k=2}^{|\Pi_\IM(\Td)|} |1 - k(k-1) \inf_{j,j'=1,\dots,k; j\ne j'} P(U_1(k) = j, U_2(k) = j')| \to 0
\end{equation}
\end{lemma}

\begin{proof}	
	Assume that $\Pi(T_k)$ is in a scattered state.  The process goes to a semi-scattered state at rate $\frac{k(k-1)m}{D} = \frac{k(k-1)\Gamma}{ND}$.  Once the blocks are in a semi-scattered state, three events can occur.  We specify the rates of these three events.
\begin{itemize}
\item the two blocks can coalesce (rate: $\frac{1}{N}$).
\item the blocks can return to a scattered state. (rate : $2m(1 - \frac{k-2}{D}) = \frac{2\Gamma}{N}(1 - \frac{k-2}{D}))$. 
\item the blocks can enter a state that is neither a scattered state nor a semi-scattered state.  (rate: $O(\frac{k^2}{ND})$).
\end{itemize}

	If the blocks return to a scattered state, the whole situation starts over.  Let the event of entering a state that is not a scattered state nor a semi scattered state be $\Theta$. A simple ratio shows
\begin{equation}  \label{E:bad_scatter}
P(\Theta) = O( \frac{k^2}{D})
\end{equation}
	Now consider $E[N_j(k \to k-1)]$.  We have,
\begin{equation} \label{E:N_j_IM_T}
E[N_j(k \to k-1)] = \mu E[T_{k-1} - T_k] 
\end{equation}
	The blocks occupy a scattered state for time with mean $\frac{ND}{\Gamma k(k-1)}$ .  Once in a semi-scattered state, outside of the event $\Theta$, the blocks either coalesce or return to a scattered state in time of order $O(N)$,  the probability of coalescing is $\frac{1}{1 + 2\Gamma} + O(\frac{k}{D})$.  Putting this all together and using (\ref{E:bad_scatter}) gives,
\begin{equation}
E[T_{k-1} - T_k] = \frac{ND}{\binom{k}{2}} \frac{1 + 2\Gamma}{2\Gamma} (1 + O(\frac{k^2}{D})).
\end{equation}
	Plugging the above expression into (\ref{E:N_j_IM_T}), summing over $k$, and taking the LPLS limit gives (\ref{E:N_j_IM}).   By the symmetry of the island model, if $\Pi(T_k)$ is in a scattered state then $P(U_1(k) = j, U_2(k) = j') = \frac{1}{k(k-1)}$.  By (\ref{E:bad_scatter}) and Lemma \ref{L:IM_scatter}, the probability of $\Pi(T_k)$ being in a scattered state over all $k$ is bounded below by $1 - O(\sum_{k=2}^nd \frac{k^2}{D})$.  This gives (\ref{E:U_IM}). 

\end{proof}

	Lemmas \ref{L:IM_IC}-\ref{L:IM_KCC} proves (\ref{E:IM_G_convergence}) in Theorem \ref{T:IM_convergence}.  We are left to specify the distribution of $\Pi_{\IM, k}(\Td)$.  As observed in \cite{Golding_1983_Genetics, Slatkin_1981_Genetics}, $B_k, b_{k,i}$ are specified by the Ewens Sampling Formula \cite{Ewens_1972_Theo_Pop_Bio}.  More precisely,  the following theorem follows from a result of Hoppe \cite{Hoppe_1984_J_Math_Bio} and our Lemmas \ref{L:IM_IC} and \ref{L:IM_scatter}.	
\begin{theorem} [Hoppe's Urn Theorem]  \label{L:Hoppe}
Let $\xi_i$ be a Bernoulli random variable with success probability $\frac{2\Gamma}{2\Gamma + i-1}$.  Assume that $\xi_2, \xi_3, \dots, \xi_n$ are independent.  Then,
\begin{equation}  \label{E:Hoppe}
B_k = 1 + \xi_2 + \xi_3 + \dots + \xi_n.
\end{equation}
and each $B_k$ is i.i.d.
\end{theorem}
Using a theorem of Donnelly and Tavare \cite{Donnelly_1986_Adv_Appl_Prob} we have the following result.
\begin{theorem}  \label{T:Tavare}
For fixed $J$,
\begin{equation}  \label{E:Tavare}
\lim_{n \to \infty} (b_{k,1}, b_{k,2}, \dots, b_{k,J}) 
	= (\Upsilon_1, \Upsilon_2, \dots, \Upsilon_J)
\end{equation}
where $\Upsilon$ is defined as in Theorem \ref{T:IM_convergence}.
\end{theorem}

	Finally we note that using Lemma \ref{L:Hoppe}, a simple computation shows $\frac{E[B_1^2]}{d} = O(\frac{\log^2(n)}{d})$.  By our assumptions on the LPLS limit of a stepping stone model coalescent we have $\frac{E[B_1^2]}{d} \to 0$.

\subsection{Stepping Stone Model and G/KC}  \label{S:SSGKC}

	This section is dedicated to the proof of Theorem \ref{T:SS_convergence}.  In \cite{Cox_1989_Ann_Prob, Cox_2002_Ann_Appl_Prob, Zahle_2005_Ann_Appl_Prob}, the authors  made significant breakthroughs in the analysis of the stepping stone model coalescent.     In this section, we draw  heavily from the theory developed in those articles, especially from the work of Zahle et al. \cite{Zahle_2005_Ann_Appl_Prob}.    Our results use the basic techniques introduced by these authors, although there are several important differences.   Zahle et al. assume that sampled individuals are initially spaced far apart, while we start with $n$ individuals in each deme. Further, Zahle et al. assume fixed $n,d$ as $N,D \to \infty$ while we take $n,d,N,D \to \infty$.  Perhaps more importantly, while Zahle et al. use an integral approach to prove their results, we use a differential approach.  
	
	We feel that the results in \cite{Cox_1989_Ann_Prob, Cox_2002_Ann_Appl_Prob, Zahle_2005_Ann_Appl_Prob} have not received the attention they deserve within the population genetics literature due to their theoretical complexity.  We hope that by providing a different approach to the theory of \cite{Cox_1989_Ann_Prob, Cox_2002_Ann_Appl_Prob, Zahle_2005_Ann_Appl_Prob}, we will encourage researchers with more applied interests to use the theory.  Below, wherever possible, we use the notation of Zahle et al.
	
		Let $\Torus$ be a two dimensional torus of width $W$ corresponding to the stepping stone model.    In the stepping stone model we may think of the blocks in $\Pi_\SS(t)$ as coalescing random walkers on $\Torus$ moving with rate $m$.  Given two random walkers on $\Torus$ let $T_0$ be the first time the two walkers occupy the same deme.  Let $t_0$ be the time at which the two walkers coalesce.  From a technical perspective it is simpler to consider a single random walker moving at rate $2m$ than two random walkers moving at rate $m$.  When we consider a single random walker we let $T_0$ be the time at which the random walker hits the origin $(0,0)$.  To consider $t_0$ we let a coalescent event occur at rate $\frac{1}{N}$ when the walker is at the origin, then $t_0$ is the time at which a coalescent event occurs.  We let $P_x^{(w)}(\Theta)$ be the probability of an event $\Theta$ for a random walker starting in deme $x$ and moving at rate $w$.  We let $p_s^{(w)}(x,y)$ be the probability that a random walker starting at $x$ and moving at rate $w$ will be in deme $y$ at time $s$.  Finally $P_x(\Theta) = P_x^{(1)}(\Theta)$ and $p_s(x,y) = p_s^{(1)}(x,y)$.
		
			 Before proceeding, we state some technical results concerning random walks on $\Torus$.  These results can be found in \cite{Cox_1989_Ann_Prob, Cox_2002_Ann_Appl_Prob}, we refer the reader to those works for the proofs. 
	 
\begin{lemma} \label{L:T_0_combined}
For $t \le \epsilon W^2 \log W$, 
\begin{equation}  \label{E:T_0_combined}
\lim_{t \to \infty} P_{(0,1)}(T_0 > t) = \frac{2\pi}{\log t}(1 + O(\epsilon)).
\end{equation}
If $|x| = o(W)$ then
\begin{equation}  \label{E:density_x}
\lim p_s(x,0) \le C \frac{1}{x^2}.
\end{equation}
If $|x| \to \infty$, $|x| = o(W)$ and $s \le x^2$ then
\begin{equation} \label{E:x_tech}
p_s(x,0) \le C \frac{\exp[-\frac{x^2}{s}]}{s}.
\end{equation}
If $t_W \to \infty$ then
\begin{equation}  \label{E:density_steady}
W^2|p_{t_W W^2}(x,y)) - \frac{1}{W^2}| \to 0.
\end{equation}
If $s \to \infty$ and $s < CW$ then 
\begin{equation}  \label{E:density_s}
\lim p_s(x,0) < \frac{C}{s}.
\end{equation}
\end{lemma}

	Set $\Td =\frac{W^2}{2m}$.  Recall that $\SR = \frac{W}{\sqrt{\log W}}$ is the minimum distance between sampled demes. 

\begin{lemma} [Independence Condition]  \label{L:SS_IC}
Let $\Theta$ be the event in which two individuals from separate sampled demes coalesce before time $\Td$.  Then,
\begin{equation}
P(\Theta) = O(\frac{(nd)^2}{\sqrt{\log W}})
\end{equation}
\end{lemma}

\begin{proof}
We can consider a single random walk moving at rate $2m$ that starts at position $x$ with $|x| > \SR$.  Let $\delta = \frac{1}{\sqrt{\log W}}$.  We compute $P_x^{(2m)}(T_0 < \Td)$ by considering the last time the walker is at the origin and rescaling time by $2m$:
\begin{align}  \label{E:T_0_T_d_original}
P_x^{(2m)} & (T_0 < \Td)  = \int_0^{W^2} ds p_s(x,0) P_{(0,1)}(T_0 > W^2 - s)
\\ \notag
	& \le \int_0^{\delta W^2} ds p_s(x,0) P_{(0,1)}(T_0 > (1-\delta)W^2) + \int_{\delta W^2}^{W^2} ds p_s(x,0) P_{(0,1)}(T_0 > W^2 - s).
\\ \notag
	& = I_1 + I_2.
\end{align}
Consider $I_1$.  Using (\ref{E:T_0_combined}) and (\ref{E:density_x}) in  the expression for $I_1$ gives,
\begin{equation}   \label{E:I_1_ind_final}
I_1 \le  \frac{\delta W^2}{x^2 \log W} = O(\delta).
\end{equation}
Now consider $I_2$.  Using (\ref{E:T_0_combined}) and (\ref{E:density_s}) gives
\begin{equation}  \label{E:I_2_independent}
I_2 \le \frac{1}{\delta W^2} \int_{\delta W^2}^{W^2} ds P(T_0 > W^2 - s) \le \frac{1}{\delta \log (W)}
\end{equation}
Combining (\ref{E:I_1_ind_final}) and (\ref{E:I_2_independent})  gives
\begin{equation}
P_x^{(2m)}(T_0 < T_d) = O(\frac{1}{\sqrt{\log W}}).
\end{equation}
Considering all possible pairs finishes the proof.

\end{proof}

\begin{lemma} [Short Scattering Phase Condition] \label{L:SS_SSPC}
\begin{equation}
P(\text{mutation before } \Td) = O(\frac{nd}{\log W}).
\end{equation}
\end{lemma}

\begin{proof}
There are at most $nd$ blocks in the time interval $[0, \Td]$.  Then,
\begin{equation}
E[\text{number of mutations in } [0,\Td]] \le \mu (nd) \Td = O(\frac{nd}{Nm}) = O(\frac{nd}{\log W}).
\end{equation}

\end{proof}

	Before demonstrating the KC Condition we prove some preliminary lemmas.  First, we show that at time $\Td$ the blocks are  far apart from one another.  Define 
\begin{equation}
\Gamma(k) = \{\pi \in \PL : |\pi| = k, \text{ if } (E_1, g_1), (E_2, g_2) \in \pi \text{ then } |g_1 - g_2| \ge  \frac{W}{(\log{W})^\frac{1}{2}}\}
\end{equation}
	
\begin{lemma}  \label{L:dispersed_at_Td}
Let $M = |\Pi_\SS(\Td)|$.  Then,
\begin{equation}
P(\Pi_\SS(\Td) \notin \Gamma(M)) = O(\frac{(nd)^2}{\log W})
\end{equation}
\end{lemma}

\begin{proof}
Given two random walkers $y_1, y_2$ starting at some arbitrary displacement $x$, by (\ref{E:density_s}) we have
\begin{equation}
P(|y_1(\Td) - y_2(\Td)| > \frac{W}{(\log W)^\frac{1}{2}}) \le \sum_{|y| \le \frac{W}{(\log W)^\frac{1}{2}}} p_{W^2}(x,y)
	= O(\frac{1}{\log W}).
\end{equation}
Considering all possible pairs gives the result.

\end{proof}

	For Lemmas \ref{L:SS_T_0}-\ref{L:SS_KC} we set $\Delta t = \epsilon (\frac{1}{2\pi m}) W^2 \log W$ and $\tilde{\Delta t} = 2m \Delta t$ where $\epsilon = \frac{1}{(\sqrt{\log W})}$.   For the sake of clarity we keep certain expressions in terms of $\epsilon$.  The results stated in Lemmas \ref{L:SS_T_0} and \ref{L:SS_smallt_0} can be found in \cite{Zahle_2005_Ann_Appl_Prob}.

\begin{lemma} \label{L:SS_T_0}
If $|x| > \frac{W}{(\log W)^\frac{1}{2}}$ then
\begin{equation}
P_x^{(2m)}(T_0 < \Delta t) = \epsilon(1 + O(e_1)),
\end{equation}
where
\begin{equation}
e_1 = \frac{\log \log W}{\sqrt{\log W}}.
\end{equation}
\end{lemma}

\begin{proof}
In \cite{Cox_1989_Ann_Prob}, Cox showed that once two blocks are sufficiently far apart, the time it takes the pair to enter the same deme is exponentially distributed with mean $\frac{W^2 \log(W)}{2 \pi m}$.  Our approach will be to divide time into intervals of size $\Delta t = \epsilon \frac{W^2 \log(W)}{2 \pi m}$.  We will show that during a time interval $\Delta t$, the probability of two blocks entering the same deme is approximately  $\epsilon$.  

By the same argument as in Lemma \ref{L:SS_IC} we have
\begin{align}
& P_x^{(2m)}  (T_0 < \Delta t)  = \int_0^{\tilde{\Delta t}} ds p_s(x,0)P_{(0,1)}(T_0 > \tilde{\Delta t} -s)
\\ \notag
	& = \int_0^{\epsilon W^2 \log \log W} ds p_s(x,0)P_{(0,1)}(T_0 > \tilde{\Delta t} -s) 
		+ \int_{\epsilon W^2 \log \log W}^{\epsilon W^2 \sqrt{\log W}} ds p_s(x,0)P_{(0,1)}(T_0 > \tilde{\Delta t} -s)
\\ \notag
	& \ \ \ + \int_{\epsilon W^2 \sqrt{\log W}}^{\tilde{\Delta t}} ds p_s(x,0)P_{(0,1)}(T_0 > \tilde{\Delta t} -s)
\\ \notag
	& = I_1 + I_2 + I_3.
\end{align} 
We first show that $I_1$ has small contribution.  Using (\ref{E:T_0_combined}),  (\ref{E:x_tech}), and (\ref{E:density_s}) we arrive at,
\begin{align}  \label{E:final_I_1}
I_1 & = O(\frac{1}{\log W})(1  + \int_{x^2}^{\epsilon W^2 \log \log W} ds \frac{1}{s}) 
 =  O(\frac{\log \log W}{\log W}) = \epsilon O(\frac{\log \log W}{\sqrt{\log W}}).
\end{align}
Now consider $I_2$.   Using (\ref{E:T_0_combined}) and (\ref{E:density_s}) gives 
\begin{equation}  \label{E:final_I_2}
I_2 = \epsilon \left( \frac{\log \log W}{\sqrt{\log W}}\right).
\end{equation}
Now consider $I_3$.  By using (\ref{E:T_0_combined}) and (\ref{E:density_steady}) some manipulation of the integral gives
\begin{equation} \label{E:final_I_3}
I_3 = \epsilon \left(1 + O(\frac{\log \log W}{\log W}) \right).
\end{equation}
Putting (\ref{E:final_I_1}), (\ref{E:final_I_2}), and (\ref{E:final_I_3}) together gives the result.  We pause to note that if we consider $\Delta t - \frac{\epsilon W^2 \log \log W}{m}$ we would have arrived at the same asymptotic result.  That is,
\begin{equation}  \label{E:T_0_special}
P_x(T_0 < \Delta t - \frac{\epsilon W^2 \log \log W}{m}) = \epsilon(1 + O(e_1)).
\end{equation}

\end{proof}

\begin{lemma}  \label{L:SS_smallt_0}
\begin{equation}
P_{(0,0)}^{(2m)}(t_0 < \Delta t) = \frac{1}{1 + \alpha} + O(e_2),
\end{equation}
where
\begin{equation}
e_2 = \frac{\log \log W}{\log W}.
\end{equation}
\end{lemma}

\begin{proof}
	  Recall, to compute $P_{(0,0)}^{(2m)}(t_0 < \Delta t)$ we consider a random walker moving at rate $2m$, with the stipulation that when the random walker is at $(0,0)$ there is a coalescent event at rate $\frac{1}{N}$.  So we may characterize the behavior of the random walker through the random variables $H, t_1, t_2, \dots, t_E, u_0, u_1, \dots, u_{E+1}$ where $H$ is the number of excursions taken by the random walker away from zero before a coalescent event occurs.  $t_1, \dots, t_H$ are the time spans of these excursions and $u_0, \dots, u_{E+1}$ are the time spans spent at the origin between excursions.  $H$ is geometric with success probability $\frac{1}{1 + 2Nm}$.  Each $u_0, u_1, \dots, u_E$ is an exponential random variables with mean of order $N$.    
	
	We first consider the distributions of the $t_i$, clearly the $t_i$ are i.i.d.  We distinguish between three types of excursions.  Set $K = \log W$ and define
\begin{flushleft}
Type I : $t_i \in [0, \frac{\Delta t}{K Nm}]$.\\
Type II : $t_i \in (\frac{\Delta t}{K Nm}, \Delta t]$.\\
Type III : $t_i > \Delta t$.\\
\end{flushleft}
By (\ref{E:T_0_combined}) we have
\begin{gather}
P(\text{Type I}) = 1 - O(\frac{1}{\log \tilde{\Delta t} - \log(K Nm)}  )   \\ \notag
P(\text{Type II}) = O(\frac{\log KNm}{(\log \tilde{\Delta t})^2}) \\ \notag
P(\text{Type III}) \to 1 - \frac{2\pi}{\log(\tilde{\Delta t})}.
\end{gather}
In the following we ignore the time contributions of the $u_i$.  Including the $u_i$ does not change the argument much, the order of the error terms stay the same, and so we drop the $u_i$ for the sake of clarity.
We first show that the probability of experiencing a Type II excursion before the coalescent event is small.  The probability of a coalescent during any given visit to the origin is $\frac{1}{1 + 2Nm} = O(\frac{1}{\log W})$.  The probability of a Type II excursion is $\frac{2Nm}{1 + 2Nm} P(\text{Type II}) = O(\frac{\log \log W}{(\log W)^2})$.  Then taking the appropriate ratio gives,  
\begin{equation}
P(\text{type II excursion before coal.}) = O(\frac{\log \log W}{\log W})
\end{equation}

	We now show that if no type II or III excursions occur then we coalesce with very high probability.  Indeed, if no Type II or III excursions occurs then we will coalesce before time $\Delta t$ if we coalesce before there are $K Nm$ Type I excursions.  The probability of not coalescing for $K Nm$ Type I excursions is 
\begin{equation}
\left( \frac{2Nm}{1 + 2Nm} P(\text{Type I}) \right)^{K Nm} = 
\left( (1 - \frac{1}{2Nm}) (1 - O(\frac{1}{\log W}) \right)^{O(\log W (Nm))} = O(\frac{1}{ W}).
\end{equation}
So up to errors of order $\frac{\log \log W}{\log W}$ we can reduce the computation of $P_{(0,0)}(t_0 < \Delta t)$ to the probability that a coalescent event occurs before a Type III excursion.  Computing the relevant ratio then gives,
\begin{equation}
P(\text{coal. before Type III}) = \frac{1}{1 + \alpha} + O(\frac{\log \log W}{\log W}).
\end{equation}
Putting all this together gives the result.  Finally we note that this result would hold if we replaced $\Delta t$ by $\frac{\epsilon W^2 \log \log W}{m}$.  That is,
\begin{equation}  \label{E:t_0_below}
P_{(0,0)}^{(2m)}(t_0 < \frac{\epsilon W^2 \log \log W}{m}) = \frac{1}{1 + \alpha} + O(e_2).
\end{equation}

\end{proof}

\begin{lemma}  [KC Condition] \label{L:SS_KC}
\begin{equation}  \label{E:N_j_SS}
\sum_{k=2}^{|\Pi_\SS(\Td)|} \sum_{j=1}^k \bigg|E[N_j(k \to k-1)] - \frac{\theta}{\binom{k}{2}} \frac{1 + \alpha}{\alpha}\bigg| \to 0
\end{equation}
For $j, j'=1,2,\dots,k$ with $j \ne j'$
\begin{equation}  \label{E:U_SS}
\sum_{k=2}^{|\Pi_\SS(\Td)|} |1 - k(k-1) \inf_{j,j'=1,\dots,k; j\ne j'} P(U_1(k) = j, U_2(k) = j')| \to 0
\end{equation}
\end{lemma}

\begin{proof}
We would like to combine Lemmas \ref{L:SS_T_0} and \ref{L:SS_smallt_0} to show that for $|x| > \frac{W}{(\log W)^\frac{1}{2}}$,
\begin{equation}  \label{E:t_0_final}
P_x^{(2m)}(t_0 < \Delta t) = \frac{\epsilon}{1 + \alpha} + \epsilon O(e_1).
\end{equation}
  By using Lemmas \ref{L:SS_T_0} and \ref{L:SS_smallt_0} we have
\begin{align}
P_x(t_0 < \Delta t)  \le P_x(T_0 < \Delta t) P_{(0,0)}(t_0 < \Delta t)
	= \frac{\epsilon}{1 + \alpha} + \epsilon  O(e_1)
\end{align}
A lower bound is provided by using (\ref{E:T_0_special}) and (\ref{E:t_0_below}):
\begin{align}
P_x(t_0 < \Delta t) & \ge P_x(T_0 < \Delta t - \epsilon W^2 \log \log W) P_{(0,0)}(t_0 < \epsilon W^2 \log \log W)
\\ \notag
			& = \frac{\epsilon}{1 + \alpha} + \epsilon O(e_1).
\end{align}
This proves (\ref{E:t_0_final}).  Up to this point we have limited ourselves to interactions of two blocks.  Now, however, we consider $\Pi_\SS(0) \in \Gamma(k)$.  First we compute $P(\Pi_\SS(\Delta t) \notin \Gamma(k-1) \cup \Gamma(k))$.  There are two ways for this two occur.  Either two blocks out of the $k$ are within $\frac{W}{\log(W)^\frac{1}{2}}$ but have not coalesced at time $\Delta t$ or two coalescent events have occurred.   We consider the first case.  Let $y_1$ and $y_2$ be random walkers moving at rate $m$ that start $x$ units apart.  Assume that if $y_1$ and $y_2$ enter the same deme then they immediately coalesce.  Let $\bar{y}_1, \bar{y}_2$ be independent random walkers that do not coalesce.  Let $\Theta$ be the event in which $y_1$ and $y_2$ do not coalesce but are within $\frac{W}{\log(W)^\frac{1}{2}}$ units of each other at time $\Delta t$.  We have the following bound,
\begin{equation}  \label{E:coupling_argument}
P(\Theta_1) \le P(|\bar{y}_1(\Delta t) - \bar{y}_2(\Delta t)| \le \frac{W}{\log(W)^\frac{1}{2}})
\end{equation}
To prove this inequality we use a coupling argument.  Couple $\bar{y}_1$ to $y_1$ and $\bar{y}_2$ to $y_2$.  By this we mean that the pairs move together.  However, if $y_1$ and $y_2$ coalesce then we decouple the two pairs and $\bar{y}_1$ and $\bar{y}_2$ begin to move independently of $y_1$ and $y_2$.  No path in $\Theta$ will experience a decoupling, so the inequality follows.  We now bound the right side of (\ref{E:coupling_argument}).
\begin{align}
P(|\bar{y}_1(\Delta t) - \bar{y}_2(\Delta t)| & \le \frac{W}{\log(W)^\frac{1}{2}}) 
\le \sum_{|z| \le \frac{W}{\log(W)^\frac{1}{2}}} p_{\Delta t}^{(2m)}(x,z)
\\ \notag
	& \le (\frac{W}{\log(W)^\frac{1}{2}})^2 \frac{C}{W^2} = O(\frac{1}{\log W}).
\end{align}
where we have used (\ref{E:density_steady}) to achieve the inequality directly above. 

	Now we consider the possibility of two coalescent events during time $\Delta t$.  By the same methods as just described, we can show that if a single coalescent event occurs at some point in time $\Delta t$, then with high probability all blocks will still be more than $\frac{W}{\sqrt{\log(W)}}$ units apart.  Then we repeat the argument and are able to show that the probability of two coalescent events is of order $O(\epsilon^2)$.  So finally we have after allowing for all possible pair combinations, 
\begin{equation} \label{E:first_e}
P(\Pi(\Delta t) \notin \Gamma(k-1) \cup \Gamma(k)) = O(\frac{k^2}{(\log W)} + k^4 \epsilon^2)
	= \epsilon O(\frac{k^4}{\sqrt{\log W}}).
\end{equation}

	From (\ref{E:t_0_final}) the probability of a coalescent event between any two blocks is $\frac{\epsilon}{1 + \alpha} + \epsilon O(e_1)$, giving
\begin{gather}  \label{E:full_asymp}
P(\Pi(\Delta t) \in \Gamma(k-1)) = \binom{k}{2} \frac{\epsilon}{1 + \alpha} + \epsilon O(k^4 e), \\ \notag
P(\Pi(\Delta t) \in \Gamma(k)) = 1 - \binom{k}{2} \frac{\epsilon}{1 + \alpha} + \epsilon O(k^4 e), \notag
\end{gather}
where
\begin{equation}
e = \frac{\log \log W}{\sqrt{\log W}}.
\end{equation}
	
	If we consider coalescent events after time $\Td$ we have,
\begin{equation}
E[T_{k-1} - T_k] = (\Delta t) \frac{1}{\binom{k}{2} \frac{\epsilon}{1 + \alpha} + O(\epsilon k^2 e)} + O(\Delta t)
\end{equation}
	This then gives,
\begin{equation}
E[N_j(k \to k-1)] = \frac{\theta}{\binom{k}{2}} (\frac{1 + \alpha}{\alpha})(1 + O(k^2 e)).
\end{equation}
	Summing over $j=1,2, \dots,k$ and then summing over $k=2,3, \dots, \Pi_\SS(\Td)$ gives
\begin{equation}  \label{E:mean_sum}
\sum_{k=2}^{|\Pi_\SS(\Td)|} \sum_{j=1}^k \left|
		E[N_j(k \to k-1)] - \frac{\theta}{\binom{k}{2}} (\frac{1 + \alpha}{\alpha}) \right|
		\le O(|\Pi_\SS(\Td)|^2 e) \le O((nd)^2 e) \to 0.
\end{equation}
Using (\ref{E:first_e}) and (\ref{E:full_asymp}) and the same argument as in Lemma \ref{L:IM_KCC} gives (\ref{E:U_SS}).

\end{proof}

	Lemmas \ref{L:SS_IC}-\ref{L:SS_KC} prove (\ref{E:SS_G_convergence}) in Theorem \ref{T:SS_convergence}.  Finally we characterize the distribution of $\Pi_{\SS,k}(\Td)$.  The result stated in Lemma \ref{L:Dist_SS} is very similar to Theorem 3 in \cite{Zahle_2005_Ann_Appl_Prob}, and our proof follows the methods introduced in Lemma \ref{L:SS_KC}, so we simply sketch the proof.
	
\begin{lemma} \label{L:Dist_SS}
\begin{equation}
\Pi_{\SS,k}(\Td) \to \Pi_\KC^{(\infty)}(\log(\frac{1 + \alpha}{\alpha})).
\end{equation}
\end{lemma}

\begin{proof}
We partition the interval $[0,\Td]$ by the points $t_k$ such that $t_k = W^{2k\rho}$ where $0 < \rho < 1$.  We  eventually select $\rho$ to optimize our error terms. We will asymptotically compute the probability of a pair coalescing in the interval $[t_k, t_{k+1}]$.  Further, we will show that at the end of this time interval, the blocks are always separated by a significant distance.  The first time interval, $[0, t_1]$, is special as we start with $n$ blocks all in the same deme.

	To make all this precise, we introduce the following notation.  If $\pi \in \PL$ then $\pi \in H_\rho^{(k)}(j)$ if $|\pi| = j$ and every pair of blocks in $\pi$ is separated by a distance of a least $\frac{W^{k\rho}}{\sqrt{\log W}}$.  Now suppose that for some $k$ such that $1 \le k \le \frac{1}{\rho}$ we have $\Pi(t_k) \in H_\rho^{(k)}(j)$.  Then by the same techniques used in Lemmas \ref{L:SS_T_0} and \ref{L:SS_smallt_0} we can show that the probability of two blocks entering the same deme is approximately $\frac{1}{k}$, and once two blocks are in the same deme, the probability of coalescing is approximately $(1 + \frac{\alpha}{k \rho})^{-1}$.  
	
	Since $|\Pi(t_k)|=j$, we have the following results
\begin{gather}
P(\Pi(t_{k+1}) \in H_\rho^{(k+1)}(j)) \to 1 - \binom{j}{2}(\frac{1}{k})(\frac{1}{1 + \frac{\alpha}{k \rho}}), \\ \notag
P(\Pi(t_{k+1}) \in H_\rho^{(k+1)}(j-1)) \to \binom{j}{2}(\frac{1}{k})(\frac{1}{1 + \frac{\alpha}{k \rho}});
\end{gather}
	For the interval $[0, t_1]$ things are a bit different as we start with  $n$ blocks that all occupy the same deme.  But in this case we can show the following
\begin{equation}
P(\Pi(t_1) \in H_\rho^{(1)}(n)) \to 1.
\end{equation}

	From the above computations, we note that up to vanishing error terms, each pair of blocks in $\Pi(t_k)$ is equally likely to coalesce in $[t_k, t_{k+1}]$.  Now we can compute the probability of no coalescent event occurring up to time $\Td$.
\begin{align}
P(\text{no coal. by }t) & = \prod_{k=1}^\frac{1}{\rho} \left( 1 - \binom{n}{2}(\frac{1}{k})(\frac{1}{1 + \frac{\alpha}{k \rho}}) \right)
\\ \notag
	& \approx \exp[- \sum_{k=1}^\frac{1}{\rho} \binom{n}{2} \frac{\rho}{(k)\rho +  \alpha}]
\\ \notag
	& \to \exp[- \binom{n}{2} \int_0^1 dt \frac{1}{t + \alpha}] = \exp[-\binom{n}{2} \log(\frac{1 + \alpha}{\alpha})].
\end{align}
This computation can be easily generalized to the probability of a coalescent event between any two time points in $[0,\Td]$.  The probabilities are recognized as precisely those of the coalescent probabilities of the Kingman coalescent run to time $\log \frac{1+\alpha}{\alpha}$.  The result then follows.

\end{proof}

	Finally we note that using Lemma \ref{L:Dist_SS} and standard Kingman coalescent results \cite{Durrett_Book_Probability_Models} we can show that $\frac{E[B_1^2]}{d} \to 0$.

\section{$F_{st}$ under Strong Mutation}  \label{S:strong}
\setcounter{equation}{0}
\setcounter{lemma}{0}

	In this section we prove Theorem \ref{T:Fst_strong}.  Recall $F_{st} = \frac{\phi_0 - \phi_1}{1 - \phi_1}\bigg|_{\phi_1 \ne 1}$.  The theorem will follow from two observations.  First $\phi_1 \to 0$ and second, $V[\phi_0] \to 0$.  More precisely the next two lemmas describe the behavior of $\phi_1$ and $\phi_0$.

\begin{lemma}  \label{L:sm_phi_1}
\begin{equation}
\lim_{r \to \infty} \LPLSL E[\phi_1 \cphio] = 0
\end{equation}
\end{lemma}

\begin{proof}
We start by considering simply $E[\phi_1]$ rather than $E[\phi_1 \cphio]$.
\begin{align}
E[\phi_1] & = \frac{1}{n^2 d^2} \sumkkjj E[\Ikkjj]
\\ \notag
	& = E[I(x_{k_1,1} = x_{k_2,1})] + O(\frac{1}{d}),
\end{align}
where $k_1 \ne k_2$.  By the definition  of a G/KC coalescent and the properties of a Kingman coalescent $x_{k_1,1}$ and $x_{k_2,1}$ coalesce at rate $1$ while a mutation occurs at rate $r$.  This gives,
\begin{equation}  \label{E:collapse_t}
E[I(x_{k_1,1} = x_{k_2,1})] = \frac{1}{1 + r} = O(\frac{1}{r}).
\end{equation}
This gives $E[\phi_1] \to O(\frac{1}{r})$.  Since
\begin{equation}
E[\phi_1] = E[\phi_1 \cphio] P(\phi_1 \ne 1) + P(\phi_1 = 1),
\end{equation}
we will have $E[\phi_1 \cphio] \to  O(\frac{1}{r})$ if we can show $P(\phi_1 = 1) \to  O(\frac{1}{r})$.  But note
\begin{equation}
P(\phi_1 = 1) \le E[\phi_1].
\end{equation}
Taking $\lim_{r \to \infty}$ finishes the proof.

\end{proof}
Now we show that $\phi_0$ approaches a deterministic value.

\begin{lemma}  \label{L:sm_phi_0}
\begin{equation}
\lim_{r \to \infty} \LPLSL \phi_0\bigg|_{\phi_1 \ne 1} = \SPL(2)
\end{equation}
\end{lemma}

\begin{proof}
We first show that $V[\phi_0] \to 0$.  
\begin{align}  \label{E:var}
V[\phi_0] = E[\left(\sum_{k=1}^d (\phi_{0,k} - E[\phi_{0,k}]) \right)^2]
\\ \notag
	& = \frac{1}{d^2} \sum_{k',k''=1, k' \ne k''}^d \text{Cov}(\phi_{0,k}, \phi_{0,k'}) + O(\frac{1}{d})
\end{align}
For $k' \ne k''$ we have the following relation
\begin{align}
\text{Cov}(\phi_{0,k}, \phi_{0,k'}) = & E[I(x_{k',1} = x_{k',2})I(x_{k'',1} = x_{k'',2})] 
\\ \notag
	&
	- E[I(x_{k',1} = x_{k',2})]E[I(x_{k'',1} = x_{k'',2})] + O(\frac{1}{n}).
\end{align}
Now we use a coupling argument introduced in \cite{Rottenstreich_2007a_Theo_Pop_Bio}. We sketch the coupling argument and direct the reader to \cite{Rottenstreich_2007a_Theo_Pop_Bio} for further details.   Let $\Pi(t)$ be a G/KC coalescent started with the following four individuals in separate blocks:  $x_{k',1}, x_{k',2}, x_{k'',1}, x_{k'',2}$.  Now define two G/KC coalescents $\Pi^{*,'}(t)$ and $\Pi^{*,''}(t)$ started with the following individuals $x_{k',1}^*, x_{k',2}^*$ and $x_{k'',1}^*, x_{k'',2}^*$ respectively  in separate blocks.  We couple $\Pi(t)$, $\Pi^{*,'}(t)$, $\Pi^{*,''}(t)$ as follows.  At the outset, the block contain each $x$ is coupled to the correspondingly indexed $x^*$.  By this we mean that the two blocks experience the same coalescent, migration, and mutation events.  If a block in $\Pi(t)$ containing a $k'$ indexed $x$ coalesces with a block containing a $k''$ indexed $x$ then we say that a decoupling has occurred.  Once a decoupling occurs, the three coalesents evolve independently. 
	Set
\begin{equation}
I = \left(I(x_{k',1} = x_{k',2}) - I(x_{k',1}^* = x_{k',2}^*)\right)
												\left(I(x_{k'',1} = x_{k'',2}) - I(x_{k'',1}^* = x_{k'',2}^*)\right)
\end{equation}
Observe,
\begin{equation}
\text{Cov}(\phi_{0,k}, \phi_{0,k'}) = E[I]
\end{equation}
Observe further, if a mutation or coalescent event occurs before the decoupling coalescent event then $I=0$.  We have,
\begin{equation}
P(\text{decouping event before mutation event}) \le 4E[I(x_{k',1} = x_{k'',1})]
\end{equation}
These two observations give
\begin{equation}  \label{E:cov}
\text{Cov}(\phi_{0,k}, \phi_{0,k'}) \le E[I(x_{k',1} = x_{k'',1})] \to O(\frac{1}{r}),
\end{equation}
where we have used (\ref{E:collapse_t}) to obtain the result directly above.  Plugging (\ref{E:cov}) into (\ref{E:var}) gives $V[\phi_0] \to O(\frac{1}{r})$.  Now note
\begin{equation}
E[\phi_0] = E[\phi_{0,1}] = E[I(x_{1,1} = x_{1,2})] + O(\frac{1}{n}).
\end{equation}
If $x_{1,1},x_{1,2}$ occupy the same block in $\Pi_\GKC(0)$ then we will have $x_{1,1} = x_{1,2}$.  Otherwise, by arguments given in Lemma \ref{L:sm_phi_1} we will have, with limiting probability $1$, $x_{1,1} \neq x_{1,2}$.  It then follows by the definition of $\SPL$ that
\begin{equation}
E[I(x_{1,1} = x_{1,2})] \to \SPL(2).
\end{equation}
Finally, recalling that $P(\phi_1 = 1) \to O(\frac{1}{r})$ from the proof of Lemma \ref{L:sm_phi_1}, leads to $V[\phi_0 \cphio] \to O(\frac{1}{r})$ and $E[\phi_0 \cphio] \to \SPL(2)$.  Taking $\lim_{r \to \infty}$ finishes the proof.

\end{proof}
Since $F_{st} \in [0,1]$, Theorem \ref{T:Fst_strong} is proved in a straightforward manner using Lemmas \ref{L:sm_phi_1} and \ref{L:sm_phi_0}.

\section{$F_{st}$ under Weak Mutation}  \label{S:weak}
\setcounter{equation}{0}
\setcounter{lemma}{0}

	The goal of this section is to prove Theorems \ref{T:bottom_tree}-\ref{T:top_tree}.  Recall that in the weak mutation setting we assume that there is a single mutation on $\Pi_\GKC(t)$.    We assume that the mutation occurs when $|\Pi_\GKC(t)| = L$. More precisely, we select a block $\Emut$ uniformly from $\Pi_\GKC(T_L)$ and mutate all individuals in $\Emut$.  Label the blocks of $\Pi_{\GKC,k}(0)$ as $E_{k,1}, E_{k,2}, \dots, E_{k,B_k}$.  We refer to any $x_{k,j} \in \Emut$ as a mutant.
	  
	Set
\begin{gather}
R_k = \sum_{j=1}^{B_k} \chi(\Emut \cap E_{k,j} \neq \emptyset), \\ \notag
R = \sum_{k=1}^d R_k.
\end{gather}
$R_k$ and $R$ are the number of blocks in $\Pi_{\GKC, k}(0)$ and $\Pi_{\GKC}(0)$ respectively that contain mutants.  Note that if a block at $t=0$ contains a single mutant, then every individual in the block must be a mutant.

	  At $t=0$, each $\F_k$ is the disjoint union of $B_k$ blocks.  Of these $B_k$ blocks, $R_k$ will contain mutants.  We refer to these $R_k$ blocks as mutant blocks.    By the symmetry of the G/KC coalescent, which it inherits from the Kingman coalescent, the mutant blocks are equally likely to be any subset of the $B_k$ blocks.  Let $\sigma(k,\cdot)$ be a random injective map from $[1,\dots,R_k]$ to $[1,\dots,B_k]$.  $\sigma(k,\cdot)$ is chosen from the uniform distribution of all such mappings.  Now define
\begin{gather}
A_k = \sum_{j=1}^{R_k} b_{k,\sigma(k,j)}. \\ \notag
\p = \frac{1}{d} \sum_{k=1}^d A_k. \\ \notag
\q = \frac{1}{d} \sum_{k=1}^d (A_k)^2.
\end{gather}
Simple algebra gives
\begin{equation}  \label{E:Fst_p}
F_{st} = \frac{\q - \p^2}{\p - \p^2}.
\end{equation}

	We will often speak of the descendants of some block $E \in \Pi_\GKC(t)$.  By this we mean all $E_i \in \Pi_\GKC(0)$ with $E_i \subset E$.  We write $\{B_i\}$ for $\{B_i\}_{i=1,\dots,d}$.  Below we let $\calA(a,b)$ be the set of all injective maps from $[1,2,\dots,a]$ to $[1,2,\dots,b]$.


\subsection{Some Preliminary Results}

	We first characterize the distributions of $\RLB$.  The LPLS limiting distribution of $\RLB$ depends on $\LPLSL L$ and $\kappa$.  We have three cases.   Define
\begin{flushleft}
$V =$ exponential random variable with mean $1$. \\
$W(z) =$ r.v. with density $(1 - \frac{1}{z}) (1 - \frac{x}{z})^{z-2}$ for $z \ge 2$ and $0\le x \le z$.\\
$G(z) =$ geometric random variable with success probability $z$.\\
\end{flushleft}

	Then we have the following result.	
\begin{lemma}  \label{L:RB_limit}
\begin{equation}
\LPLSL \frac{RL}{B}\bigg|_B = \bigg\{
\begin{array} {cc}
V & \text{if }  \kappa = 0, \LPLSL L \to \infty, \\
W(L) & \text{if } \kappa = 0, \LPLSL L < \infty,  \\
\kappa (G(\kappa)+1) & \text{if }  \kappa \ne 0.
\end{array}
\end{equation}
\end{lemma}

\begin{proof}
Before proving the three cases we show that $\frac{B}{E[B_1] d} \to 1$.  Indeed, by our assumption of $\frac{E[B_1^2]}{d} \to 0$ in the LPLS limit we have 
\begin{equation}  \label{E:collapse_LLN}
V[\frac{B}{d} - E[B_1]] = \frac{1}{d} V[B_1] \le \frac{E[B_1^2]}{d} \to 0.
\end{equation}
We can then conclude
\begin{equation}
\kappa = \LPLSL \frac{L}{E[B_1] d} = \LPLSL \frac{L}{B} \frac{B}{d E[B_1]} = \LPLSL \frac{L}{B}.
\end{equation}

Let $j_1, j_2, \dots, j_L$ be the number of descendants from each block in $\Pi_\GKC(T_L)$.  A standard result, see for instance \cite{Durrett_Book_Probability_Models}, is
\begin{equation}
P(j_1, j_2, \dots, j_L \ | \ B) = \frac{1}{\binom{B-1}{L-1}}
\end{equation}
By symmetry we may set $R=j_1$.  Then elementary combinatorics gives
\begin{equation}
P(R \ | \ B) = \frac{\binom{B-R-1}{L-2}}{\binom{B-1}{L-1}}.
\end{equation}

	Now we consider the three cases stated in the lemma. For simplicity of notation let $Z = \frac{RL}{B}\bigg|_B$.   First take $\kappa = 0, L \to \infty$.  In this case since $\kappa = 0$ we have $\frac{L}{B} \to 0$.  
\begin{align}  \label{E:base_comb}
\LPLSL P(a \le Z \le b) & = \LPLSL \sum_{R=\frac{aB}{L}}^{\frac{bB}{L}} \frac{\binom{B-R-1}{L-2}}{\binom{B-1}{L-1}}.
\\ \notag
	& = \LPLSL \sum_{R=\frac{aB}{L}}^{\frac{bB}{L}} \frac{L-1}{B-1} (1 - \frac{R}{B-1})^{L-2} E(L,R,B).
\end{align}
where
\begin{equation}
E(L,R,B) = \frac{\prod_{j=1}^{L-3} 1 - \frac{j}{B-R-1}}{\prod_{j=1}^{L-2} 1 - \frac{j}{B-1}}.
\end{equation}
A standard argument then shows, since $\frac{L}{B} \to 0$ and $L \to \infty$ that,
\begin{equation}
\LPLSL P(a \le Z \le b) \to \int_a^b dx \exp[-x].
\end{equation}
In the case $\kappa = 0$, $\LPLSL L < \infty$, we can use (\ref{E:base_comb}) to show that $Z \to W(L)$.  Now consider the case $\kappa > 0$.    Taking $\epsilon > 0$, 
\begin{equation}
\LPLSL P(\kappa j - \epsilon \le Z \le \kappa j + \epsilon) = \LPLSL P(R = k) = \frac{\binom{B-R-1}{L-2}}{\binom{B-1}{L-1}}
\end{equation}
Now expanding the binomials directly above and taking the LPLS limit gives that $R-1$ goes to a geometric random variable with success probability $\kappa$.  The lemma follows.

\end{proof}

	Lemma \ref{L:RB_limit} shows that $\RLB$ has three different limits depending on the scaling of $L$ that we choose.  In each case we want to compute the LPLS limit of the mean and variance of $\p$ and $\q$ conditioned on  $\RLB$.  This however is technically cumbersome because prior to taking the LPLS limit, $\RLB$ is discrete.  Furthermore, if $\kappa > 0$, the LPLS limit of $\RLB$ is discrete.  To deal with all three limits of $\RLB$ simultaneously, and to avoid unneeded technical difficulties we condition not on $\RLB$, but on the event $\RLB \in \Ih$ for certain sets $\Ih$.  More precisely, let $\epsilon > 0$, then set
\begin{equation}
\Ih = \bigg\{
\begin{array} {cc}
[h\epsilon, (h+1)\epsilon) \text{ for } h=0,1,2,\dots & \text{if } \kappa = 0, \LPLSL L \to \infty. \\ 
\text{[}h\epsilon,(h+1)\epsilon) \text{ for } h=0,1,2,\frac{L}{\epsilon} & \text{if } \kappa = 0, \LPLSL L < \infty. \\ 
(h \kappa - \epsilon, h \kappa + \epsilon) \text{ for } h=0,1,2,\dots & \text{if } \kappa \ne 0.
\end{array}
\end{equation}

\begin{lemma}  \label{L:r_k_limit}
Let $i$ be a positive integer with $i \le B_k$.  Then,
\begin{equation}  \label{E:r_k_bound}
\LPLSL E(\binom{R_k}{i} \cRallB)  \le  \binom{B_k}{i} (\RB)^i(1 + O(\frac{B_k^2}{B}).
\end{equation}
\begin{equation}  \label{E:r_k_limit_1}
\LPLSL E(\binom{R_k}{i} \cRallB)  =  \binom{B_k}{i} (\RB)^i(1 + O(\frac{B_k^2}{B-R} + \frac{R_k^2}{R})).
\end{equation}
For $k \ne k'$ and $i,i'$ positive fixed integers,
\begin{equation}  \label{E:r_k_limit_2}
\LPLSL E[R_k^i R_{k'}^{i'} \cRallB) = \LPLSL E[R_k^i \cRallB] E[R_{k'}^{i'} \cRallB] ( 1 + 
				  O(\frac{B_k^2 + B_{k'}^2}{B-R} + \frac{R_k^2 + R_{k'}^2}{R}) ).
\end{equation}
\end{lemma}

\begin{proof}
We demonstrate (\ref{E:r_k_bound}) and (\ref{E:r_k_limit_1}), the proof of (\ref{E:r_k_limit_2}) is similar.  We choose $R$ mutant blocks out of a total of $B$ possible blocks.  Each collection of $R$ choices is equally likely, so we have
\begin{equation}
P(R_k \cRallB) = \frac{\binom{B_k}{R_k}\binom{B-B_k}{R-R_k}}{\binom{B}{R}}
\end{equation}
From the relation directly above one can show
\begin{equation}
P(R_k \cRallB)	\le  \binom{B_k}{R_k} \left(\RB\right)^{R_k} \left(1 - \RB\right)^{B_k - R_k} (1  
		+ O(\frac{B_k^2}{B})). 
\end{equation}
and
\begin{equation}
P(R_k \cRallB)
			= \binom{B_k}{R_k} \left(\RB\right)^{R_k} \left(1 - \RB\right)^{B_k - R_k} (1 
		+ O(\frac{R_k^2}{R} + \frac{B_k^2}{B-R} + \frac{R_k^2}{B-R})). 
\end{equation}
These two relations give (\ref{E:r_k_bound}) and (\ref{E:r_k_limit_1}) respectively.

\end{proof}

The following lemma will be used to control the error expression produced in Lemma \ref{L:r_k_limit}.

\begin{lemma}  \label{L:control_B_k_error}
If $\kappa = 0$ and $\LPLSL L < \infty$ assume $h \ne \frac{L}{\epsilon}, \frac{L}{\epsilon} - 1$.  
\begin{equation}
\LPLSL E[ \frac{B_k^2}{B-R} + \frac{R_k^2}{R}  \cIh] = 0
\end{equation}
\end{lemma}

\begin{proof}
Let $H =  \frac{B_k^2}{B-R} + \frac{R_k^2}{R}$.
We have,
\begin{align}
H =
	& \frac{B_k^2}{B} \left(\frac{1}{1 - \RLB(\frac{1}{L})} \right) + \frac{R_k^2}{R}. 
\end{align}
From (\ref{E:r_k_bound}) we have 
\begin{equation}
E[\frac{R_k^2}{R} \ | \ R,B] \le \frac{B_k^2}{R} (\frac{R}{B})^2 (1 + O(\frac{B_k^2}{B})) = O(\frac{B_k^2}{B}).
\end{equation}
By our assumptions on $h$ we have $\limsup \RLB \frac{1}{L} < 1$.  So we arrive at,
\begin{equation}
E[H \cIh] \le O(E[\frac{B_k^{2}}{B} \cIh]).
\end{equation}
We now write out the conditional expectation explicitly.  Without loss of generality we take $k=1$.
\begin{align}  \label{E:cond_B_k_B}
E[&\frac{B_1^{2}}{B} \cIh] 
 = \frac{\sum_{B_1} P(B_1) B_1^{2} \sum_{B_2,\dots,B_d} P(B_2,\dots,B_d) \sum_{\RLB \in \Ih} P(\frac{RL}{B} \ | \ B) \frac{1}{B}}{P(\RLB \in \Ih)}
\end{align}
But now we note that by Lemma \ref{L:RB_limit}, $P(\RLB \in \Ih)$ is asymptotically independent of $B$.  So using (\ref{E:cond_B_k_B}) we have
\begin{equation}
\LPLSL E[H \cIh] = \LPLSL E[\frac{B_1^{2}}{B}] = \LPLSL \frac{E[B_1^{2}]}{dE[B_1]}  = 0,
\end{equation}

\end{proof}

	We will need to compute the moments of products of $b_{k,j}$.  The following lemma shows that such moments can be expressed in terms of the scattering probabilities.  In general we will be computing products of $b_{k,j}$ for uniformly selected $j$ over $1,\dots, B_k$.  To make this precise let $I$ be a positive integer and let $\gamma$ be a random element of  $\calA(I,B_k)$ under the uniform distribution.  We have the following lemma.
	
\begin{lemma}  \label{L:scat_prob}
Let $I, j_1, j_2, \dots, j_I$ be fixed positive integers with each $j_i$ unique. Set $J = j_1 + j_2 + \dots + j_I$.  Then for $B_k > I$, $J < n$,
\begin{equation}
\LPLSL I! \binom{B_k}{I} \binom{J}{j_1, j_2, \dots, j_I} 
E[\prod_{i=1}^I b_{k,\gamma(i)}^{j_i} \ | \ B_k] = \SPL(j_1,j_2,\dots,j_I)\bigg|_{B_k}.
\end{equation}
\end{lemma}

\begin{proof}
If we sample $J$ individuals from $\F_k$, then $\SP(j_1,j_2,\dots,j_I)$ is the probability that the blocks $E_{k,1}, E_{k,2}, \dots, E_{k,B_k}$ partition the $J$ individuals into $I$ sets of size $j_1, j_2, \dots, j_I$.  Taking ordering into account, there are $J! \binom{n}{J}$ ways to sample $J$ individuals from $\F_k$.  There are $\frac{(nb_{k,h})!}{(nb_{k,h}-j_i)!}$ ways to assign $j_1$ individuals to block $E_{k,h}$.  With this in mind, if we consider all possible combinations, we arrive at 
\begin{equation}
\SP(j_1,j_2,\dots,j_I)\bigg|_{b_{k,1}, b_{k,2},\dots,b_{k,B_k}} = \frac{1}{J! \binom{n}{J}}
	\sum_{\gamma \in \calA(I,B_k)} \binom{J}{j_1, j_2, \dots, j_I} \prod_{i=1}^I
			\frac{(nb_{k,\gamma(i)})!}{(nb_{k,\gamma(i)}-j_i)!},
\end{equation}
Since we fix $J$, taking the LPLS limit gives the following asymptotics
\begin{equation}
\SPL(j_1,j_2,\dots,j_I)\bigg|_{b_{k,1}, b_{k,2},\dots,b_{k,B_k}} = \LPLSL 
	\sum_{\gamma \in \calA(I,B_k)}   \binom{J}{j_1, j_2, \dots, j_I} \prod_{i=1}^I b_{k,\gamma{i}}^{j_i}
\end{equation}
Noting that $\sum_{\gamma \in \mathcal{A}_J} = \frac{B_k!}{((B_k-I)!}$ leads to
\begin{equation}
\SPL(j_1,j_2,\dots,j_I)\bigg|_{b_{k,1}, b_{k,2},\dots,b_{k,B_k}} = \LPLSL I! \binom{B_k}{I}
	\sum_{\gamma \in \calA(I,B_k)}  P(\gamma) \binom{J}{j_1, j_2, \dots, j_I} \prod_{i=1}^I b_{k,\gamma{i}}^{j_i}
\end{equation}
If we now condition $\Theta$ over $B_k$ rather than $b_{k,1}, b_{k,2}, \dots, b_{k,B_k}$ we have,
\begin{equation}
\SP(j_1,j_2,\dots,j_I)\bigg|_{B_k} = \LPLSL 
	I! \binom{B_k}{I} \binom{J}{j_1, j_2, \dots, j_I} E[\prod_{i=1}^I b_{k,\gamma(i)}^{j_i} \ | \ B_k].
\end{equation}

\end{proof}

	Finally, we show that the distribution of $\SPL$ depends very weakly on $B_k$.  

\begin{lemma}  \label{L:control_Theta_error}
With the notation and conditions of Lemma \ref{L:scat_prob},
\begin{equation}
E[\SPL(j_1, \dots, j_I)\bigg|_{B_k} \cIh] = \SPL(j_1, \dots, j_I).
\end{equation}
\end{lemma}

\begin{proof}
The proof of this lemma is very similar to that of Lemma \ref{L:control_B_k_error}.  The existence of a limit for $P(\Ih)$ allows us to remove the conditional dependence on $\Ih$.

\end{proof}

\subsection{$\p$}  \label{S:p}

		Now we consider $\p$  conditioned on $\Ih$.  
\begin{lemma}  \label{L:mean_p}
If $\kappa = 0$ and $\LPLSL L < \infty$ assume $h \ne \frac{L}{\epsilon}, \frac{L}{\epsilon} - 1$.  
\begin{equation}
\LPLSL E[L\p  \cIh] \in \Ih
\end{equation}
\end{lemma}

\begin{proof}
Using the fact that $R_k, b_{k, \sigma(k,j)}$ are independent when conditioned on $B_k$, we have
\begin{align}  \label{E:mean_p_1}
E[\p \cRallB] & = E[A_k \cRallB] = E[\sum_{j=1}^{R_k} b_{k,\sigma(k,j)} \cRallB]
\\ \notag
	& = E[\sum_{j=1}^{R_k} E[b_{k,\sigma(k,j)} \ | \ B_k] \cRallB]
    = E[R_k E[b_{k,\sigma(k,1)} \ | \ B_k] \cRallB]
\end{align}
Applying Lemma \ref{L:scat_prob} with $J = I = 1$, noting $\SP(1) = 1$, and then applying Lemma \ref{L:r_k_limit} leads to
\begin{align}  \label{E:mean_p_2}
E[L\p \cRallB] = \frac{L}{B_k} E[R_k  \cRallB] = \RLB + \RLB O(\frac{B_k^2}{B-R} + \frac{R_k^2}{R}).
\end{align}
Now if we condition both sides of the above equation with respect to $\Ih$ and
apply Lemma \ref{L:control_B_k_error} we arrive at the statement of the proof.

\end{proof}

	Having computed the conditional mean of $L \p$ on $\Ih$, we now consider the conditional variance.
	
\begin{lemma}  \label{L:var_p}
If $\kappa = 0$ and $\LPLSL L < \infty$ assume $h \ne \frac{L}{\epsilon}, \frac{L}{\epsilon} - 1$.  Then,
\begin{equation}
\LPLSL V[L \p \cIh] \le O(\lambda).
\end{equation}
\end{lemma}

\begin{proof}
We start by considering $E[L^2\p^2 \cRallB]$.  
\begin{equation}  \label{E:p_2_main}
E[L^2\p^2 \cRallB] = \frac{L^2}{d^2} \sum_{k,k'=1}^d E[A_k A_{k'} \cRallB].
\end{equation}
So we need to compute $E[A_k^2 \cRallB]$ and $E[A_k A_{k'} \cRallB]$ for $k \ne k'$.  Starting with $E[A_k^2 \cRallB]$ and expanding out $A_k$ gives
\begin{align}  \label{E:A_k_2}
E[A_k^2 \cRallB] = & E[R_k (R_k-1) E[b_{k,\sigma(k,1)} b_{k,\sigma(k,2)} \ | \ B_k] \cRallB]
\\ \notag
			& + E[R_k E[b_{k,\sigma(k,1)}^2  \ | \ B_k] \cRallB]
\end{align}
Using Lemma \ref{L:scat_prob} gives,
\begin{gather}  \label{E:various_b_moments}
B_k(B_k-1)  E[b_{k,\sigma(k,1)} b_{k, \sigma{k,2}} \ | \ B_k] \to \SPL(1,1)\bigg|_{B_k}. \\ \notag
B_k E[b_{k,\sigma(k,1)}^2 \ | \ B_k] \to \SPL(2)\bigg|_{B_k}.
\end{gather}
Plugging (\ref{E:various_b_moments}) into (\ref{E:A_k_2}) and using Lemma \ref{L:r_k_limit} gives
\begin{align}  \label{E:A_k_2_final}
\LPLSL & E[A_k^2 \cRallB] 
\\ \notag  
	& = \LPLSL \SPL(1,1)\bigg|_{B_k} \frac{1}{B_k(B_k-1)} E[R_k (R_k-1) \cRallB]
			+ \SPL(2)\bigg|_{B_k} \frac{1}{B_k} E[R_k  \cRallB]
\\ \notag
    & = \SPL(1,1)\bigg|_{B_k} (\RB)^2
			+ \SPL(2)\bigg|_{B_k} \RB + (\RB) O(\frac{B_k^2}{B-R} + \frac{B_k^2}{R})
\\ \notag
		& = (\RB)^2 + \SPL(2)\bigg|_{B_k} \RB(1 - \RB) + (\RB)O(\frac{B_k^2}{B-R} + \frac{B_k^2}{R});
\end{align}
where we have used the relation $\SPL(2)\bigg|_{B_k} = 1 - \SPL(1,1)\bigg|_{B_k}$ to arrive at the final equality.  

	Now we turn to $E[A_k A_{k'} \cRallB]$ for $k \ne k'$.  An argument similar to the one just finished for $E[A_k^2 \cRallB]$ gives
\begin{equation}  \label{E:A_k_k_final}
E[A_k A_{k'} \cRallB] = (\RB)^2 + (\RB)^2 O(\frac{B_k^2}{B-R} + \frac{B_k^2}{R})
\end{equation}
Plugging (\ref{E:A_k_2_final}) and (\ref{E:A_k_k_final}) into (\ref{E:p_2_main}) gives
\begin{equation}
E[L^2\p^2 \cRallB] \to (\RLB)^2 + \frac{L}{d} \SPL(2)\bigg|_{B_k} (\RLB)(1 - \RB) + (\RLB)^2 O(\frac{B_k^2}{B-R} + \frac{B_k^2}{R}).
\end{equation}
Using (\ref{E:mean_p_2}) we can express the variance as follows,
\begin{equation}
V[L\p \cRallB] \to \frac{L}{d} \SPL(2)\bigg|_{B_k} (\RLB)(1 - \RB) + (\RLB)^2 O(\frac{B_k^2}{B-R} + \frac{B_k^2}{R}).
\end{equation}
We then condition on $\Ih$ and use Lemmas \ref{L:control_B_k_error} and \ref{L:control_Theta_error} to arrive at
\begin{equation}
\LPLSL V[L \p \cIh] = \LPLSL \frac{L}{d} \SPL(2)\bigg|_{B_k} E[(\RLB)(1 - (\RLB)\frac{1}{L}) \cIh] \le O(\lambda).
\end{equation}

\end{proof}

\subsection{$\q$}  \label{S:q}
	As we did in the previous section for $\p$, in this section we compute the mean and variance of $\q$.  
	
\begin{lemma}
If $\kappa = 0$ and $\LPLSL L < \infty$ assume $h \ne \frac{L}{\epsilon}, \frac{L}{\epsilon} - 1$.  Let $x \in \Ih$.  Then,
\begin{equation}
\LPLSL E[L\q \cIh] = \frac{x^2}{L} + \SPL(2) (x)(1 - \frac{x}{L}) + O(\epsilon).
\end{equation}
\end{lemma}

\begin{proof}
We have $\q = \frac{1}{d} \sum_{k=1}^d A_k^2$.  The result then follows from (\ref{E:A_k_2_final}).

\end{proof}

\begin{lemma}  \label{L:var_q}
If $\kappa = 0$ and $\LPLSL L < \infty$ assume $h \ne \frac{L}{\epsilon}, \frac{L}{\epsilon} - 1$.  
\begin{equation}
\LPLSL V[L\q \cIh] = O(\lambda).
\end{equation}
\end{lemma}

\begin{proof}
We sketch the proof as it is very similar in technique to Lemmas \ref{L:mean_p} and \ref{L:var_p}.
Using (\ref{E:r_k_limit_2}) it is not hard to show that for $k \ne k'$,
\begin{equation}
E[A_k^2 A_{k'}^2 \cRallB] = E[A_k^2 \cRallB] E[A_{k'}^2 \cRallB] + O(\frac{B_k^2 + B_{k'}^2}{B-R} + \frac{B_k^2 + B_{k'}^2}{R}).
\end{equation}
Since asymptotically the $A_k$ are uncorrelated, the variance of $L\q$ reduces to the variance of $LA_k^2$.  Ignoring error terms this gives,
\begin{equation}  \label{E:V_q}
V[L\q \cRallB] = \frac{L^2}{d} \sum_{k=1}^d (E[A_k^4 \cRallB] - E[A_k^2 \cRallB]^2).
\end{equation}
From (\ref{E:A_k_2_final}) we have (again ignoring error terms)
\begin{equation}  \label{E:A_k_2_2}
E[A_k^2 \cRallB]^2 = 
	\left((\RB)^2 + \SPL(2)\bigg|_{B_k} \RB(1 - \RB)\right)^2.
\end{equation}
Using Lemmas \ref{L:r_k_limit} and \ref{L:scat_prob} as we did in Lemma \ref{L:var_p} gives
\begin{align}  \label{E:A_k_4}
E[A_k^4 \cRallB] & = O(\frac{R}{B}).
\end{align}
Plugging (\ref{E:A_k_2_2}) and (\ref{E:A_k_4}) into (\ref{E:V_q}) gives
\begin{equation}
V[L\q \cRallB] = O((\RLB) \frac{L}{d}) = O(\lambda).
\end{equation}

\end{proof}

\subsection{Limit of $F_{st}$}
We can now put together the results of sections \ref{S:p} and \ref{S:q} to demonstrate Theorems \ref{T:bottom_tree}-\ref{T:top_tree}.  We start by proving Theorem \ref{T:bottom_tree}.

\begin{proof}[Theorem \ref{T:bottom_tree}]
We will consider $F_{st}$ conditioned on $\RLB \in \Ih$ as $\epsilon \to 0$.  
All the lemmas developed in sections \ref{S:p} and \ref{S:q} include the assumption that if $\kappa = 0$ and $\LPLSL L < \infty$ then also $h \ne \frac{L}{\epsilon}, \frac{L}{\epsilon} - 1$.  But as $\epsilon \to 0$, $P(\RLB \in \Ih) \to 0$ for these values of $h$.  With this in mind, for the rest of this proof we assume that $h$ does not take on these excluded values.

Rewriting (\ref{E:Fst_p}) gives
\begin{equation}
F_{st}\bigg|_{\Ih} = \frac{L\q - (L\p)^2\frac{1}{L}}{L\p - (L\p)^2\frac{1}{L}}\bigg|_{\Ih}.
\end{equation}
Now note that by Lemmas \ref{L:mean_p}-\ref{L:var_q}, since $\lambda = \LPLSL \frac{L}{d} = 0$, the means of $L\p$ and $L\q$ go to non-zero limits while the variance collapses.  If we plug in the mean values for $L\p$ and $L\q$ we arrive at
\begin{equation}
\LPLSL F_{st}\bigg|_{\Ih} = \SPL(2) + O(\epsilon).
\end{equation}
Since the limit is independent of $h$ and since $F_{st}$ is bounded a dominated convergence theorem argument shows $F_{st} \to \SPL(2)$.

\end{proof}

The proofs of Theorems \ref{T:middle_tree} and \ref{T:top_tree} are harder and require some preparation.      The following lemma simplifies the expression for $F_{st}$.

\begin{lemma}  \label{L:p_q_limit}
For $\lambda > 0$,  
\begin{equation}
\LPLSL F_{st} = \LPLSL \frac{\q}{\p}
\end{equation}
\end{lemma}
\begin{proof}
We have
\begin{equation}  \label{E:first_Fst_limit}
F_{st} = \frac{\q - \p^2}{\p - \p^2} = \frac{\q}{\p(1 - \p)} + \frac{\p}{1-\p}.
\end{equation}
Now note that by Lemmas \ref{L:RB_limit} and \ref{L:mean_p}, $E[L\p] \to c > 0$.  Since $L = \lambda d \to \infty$, we have $\p \to 0$.  Using this observation in (\ref{E:first_Fst_limit}) finishes the proof.

\end{proof}

	Before stating the next lemma we define the random variables $\brand(z)$ and $\srand$.  $\srand$ is given by the following distribution.  For $i=1,2,3,\dots$,
\begin{equation}
P(\srand = i) = \frac{i P(B_1=i)}{E[B_1]}.
\end{equation}
Now we define $\brand$.  Let $\eta$ be a uniform random variable on $\{1,2,\dots,z\}$.  Then for $a,b \in [0,1]$
\begin{equation}
P(\brand(z) \in [a,b]) = P(b_{1,\eta} \in [a,b] | B_1 = z).
\end{equation} 
So $\brand(z)$ is the relative size of a block uniformly chosen from $z$ blocks that partition $\F_1$.  The following lemma expresses $F_{st}$ in terms of $\brand(\srand)$.

\begin{lemma}  \label{L:pq_limits}
Assume $\lambda > 0$.  Define
\begin{equation}
Y = \bigg\{
\begin{array}{cc}
\lceil \frac{VE[B_1]}{\lambda}\ \rceil & \text{ if } \kappa = 0\\
G(\kappa) + 1 & \text{ if } \kappa \ne 0.
\end{array}
\end{equation}
Let $\brand_1, \brand_2, \dots$ be independent versions of $\brand$ and $\srand_1, \srand_2, \dots$ be independent versions of $\srand$.  Then,
\begin{equation} \label{E:p_lim}
\LPLSL \p = \LPLSL \sum_{j=1}^Y \brand_j(\srand_j)
\end{equation}
\begin{equation} \label{E:q_lim}
\LPLSL \q = \LPLSL \sum_{j=1}^Y \brand_j^2(\srand_j)
\end{equation}
\begin{equation}  \label{E:Fst_lim}
\LPLSL F_{st} = \LPLSL \frac{\sum_{j=1}^Y \brand_j^2(\srand_j)}
	{\sum_{j=1}^Y \brand_j(\srand_j)}
\end{equation}
\end{lemma}

\begin{proof}
We start by considering $\p$ and $\q$ conditioned on $B$.  To simplify our index notation let $b_1, b_2, \dots, b_B$ be  some ordering of the collection $b_{k,j}$ for $k=1,\dots,d$ and $j=1,\dots,B_k$.  Let $\zeta(k)$ be the sample deme associated with $b_k$.   That is, if $b_h$ is the reindexed version of $b_{k,j}$ then $\zeta(h) = k$.  

	If we condition on $B$, $\p$ and $\q$ are specified by choosing $R$ blocks out of the $B$ possible blocks, where each subset of $R$ is equally likely.    Then we can specify $\p$ through (recall the definition of $\calA$ immediately after (\ref{E:Fst_p}))
\begin{equation}
\p = \frac{1}{d} \sum_{h=1}^R b_{f(h)},
\end{equation}
where f is a random element of $\calA(R,B)$ under the uniform distribution.  Now we let $g_1, \dots, g_R$ be uniform r.v. on $[1,2,\dots,B]$.  Then we claim $\LPLSL \p = \LPLSL \frac{1}{d} \sum_{h=1}^R b_{g_h}$.  We do this through a coupling argument.  We select $g_1, g_2, \dots, g_R$.  If each one is different, then we define $f(h) = g_h$.  If some $g_i = g_{i'}$, then we select $f$ according to its (uniform) probability distribution.  We would like to show that the probability of uncoupling goes to zero in the LPLS limit.  
\begin{equation}
P(\text{uncoupling} \ | \ R,B) \le \binom{R}{2} \frac{1}{B^2} \le (\RLB) \frac{1}{L^2}.
\end{equation}
Lemma \ref{L:RB_limit} shows that $\LPLSL \RLB\bigg|_B$ exists and is independent of $B$ and since $L \to \infty$ we have
\begin{equation}
P(\text{uncoupling} \ | \ B) \to 0.
\end{equation}
which implies
\begin{gather}  \label{E:end_step_1}
\LPLSL \p = \LPLSL \sum_{j=1}^R b_{g_j}, \\ \notag
\LPLSL \q = \LPLSL \sum_{j=1}^R b_{g_j}^2.
\end{gather}

	Now we show that we may replace the $R$ by $Y$.  We restrict our attention to the case $\kappa = 0$ and consider $\p$ only.  The case $\kappa \ne 0$ is much simpler since $R$ converges to a geometric distribution, and the analysis of $\q$ is similar to that of $\p$.  We first show that we can replace $R$ by $Y' = \lceil(\RLB) \frac{E[B_1]}{\lambda} \rceil$.  
\begin{align}  \label{E:Y_prime}
E[|\sum_{j=1}^R b_{g_j} - \sum_{j=1}^{Y'} b_{g_j}|] 
& \le E[\lceil |Y'-R| \rceil]E[b_g] \le E[\lceil|Y'-R|\rceil]
	= E[\lceil (\RLB)(\frac{B}{L} - \frac{E[B_1]}{\lambda}) \rceil]
\\ \notag
	& = E[\lceil (\RLB)\frac{1}{\lambda}(\frac{B}{d} - E[B_1]) \rceil] \to 0.
\end{align}
Finally we would like to show that we can replace $Y'$ by $Y$.  To do this we recall that we have split $[0,\infty)$ into intervals $\Ih$.  By Lemma \ref{L:RB_limit}, $P(\RLB \in \Ih) \to P(V \in \Ih)$.  So we have
\begin{equation}
E[\big|\sum_{j=1}^{Y'} b_{g_j} - \sum_{j=1}^Y b_{g_j}\big| \ | \ V, \RLB \in \Ih] \le E[\sum_{j=1}^{\epsilon \frac{E[B_1]}{\lambda}} b_{g_j}] \le \epsilon \frac{E[B_1]}{\lambda} E[b_g].
\end{equation}
Now note that $E[b_g \ | \ B] = E[\frac{1}{B} \sum_{j=1}^B b_j \ | \ B] = \frac{d}{B}$.  Plugging this observation into the inequality directly above gives
\begin{equation}
E[\big|\sum_{j=1}^{Y'} b_{g_j} - \sum_{j=1}^Y b_{g_j}\big| \ | \ V, \RLB \in \Ih] \le E[\frac{\epsilon}{\lambda} \frac{E[B_1]d}{B}] \to \epsilon.
\end{equation}
Now taking $\epsilon$ to zero shows that we can replace $Y'$ by $Y$.  

	Now we would like to show
\begin{equation}
\LPLSL \sum_{j=1}^Y b_{g_j} = \LPLSL \sum_{j=1}^Y \brand_j(\srand_j)
\end{equation}
To do this we compute the LPLS limit of the characteristic function of $\sum_{j=1}^Y b_{g_j}$, $\psi(\nu)$.  Recall that $\zeta(g)$ is the sample deme to which $b_g$ is associated.
\begin{equation}  \label{E:psi}
\psi(\nu) = E[\exp[i \nu b_g]]^Y = \left(\sum_{j=1}^\infty P(B_{\zeta(g)} = j) E[\exp[i \nu b_g] | B_{\zeta(g)} = j]\right)^Y.
\end{equation}
If we condition on $B_1, B_2, \dots, B_d$ then 
\begin{equation}  \label{E:P_B_zeta}
P(B_{\zeta(g)} = j \ | \ \{B_i\}) = \frac{\sum_{k=1}^d \xi(B_k = j)j}{B} = \frac{\frac{1}{d} \sum_{k=1}^d \xi(B_k = j)j}{\frac{1}{d} B}
\end{equation}
Now note that $\xi(B_k = j)$ are i.i.d so by law of large numbers $\frac{1}{d} \sum_{k=1}^d \xi(B_k = j)j \to P(B_1 = j)j$, while $\LPLSL \frac{B}{dE[B_1]} = 1$.  So defining $\delta_j$ through the following relation
\begin{equation}  \label{E:def_delta}
\frac{\sum_{k=1}^d I(B_k = j)j}{\frac{1}{d} B} =  \frac{P(B_1 = j)j}{E[B_1]}(1 + \delta(j)),
\end{equation}
and $\delta(j) \to 0$.  Plugging (\ref{E:def_delta}) into (\ref{E:P_B_zeta}) and then plugging the result into (\ref{E:psi}) gives
\begin{align}  \label{E:psi_2}
\psi(\nu) & = \left(\sum_{j=1}^\infty (\frac{P(B_1 = j)j}{E[B_1]}(1 + \delta(j))) E[\exp[i \nu b_g] | B_{\zeta(g)} = j]\right)^Y
\\ \notag
	& = \left(\frac{1}{E[B_1]} E[B_1 (1 + \delta(B_1)) \exp[i \nu b_g] \ | \ \zeta(g) = 1]\right)^Y
\end{align}
We now expand $\exp[i \nu b_g]$ in Taylor series.  From Lemma \ref{L:scat_prob}, we have the following relation for the moments of $b_g$, for $k>1$.
\begin{equation}  \label{E:moments_b_g}
E[b_g^k \ | \ \zeta(g) = 1, B_1] =  E[\frac{1}{B_1} \SP(k) \ | \ B_1].
\end{equation}
Plugging (\ref{E:moments_b_g}) into (\ref{E:psi_2}) gives
\begin{equation}
\psi(\nu)  
	= (1 + \frac{1}{E[B_1]} \sum_{k=1}^\infty \frac{(i\nu)^k}{k!} E[(1 + \delta(B_1)) \SP(k)])^Y
\end{equation}
Now recall $Y = \frac{VE[B_1]}{\lambda}$ and notice that $E[B_1] \to \infty$ since $\kappa \to 0$.  These facts lead to
\begin{equation}
\LPLSL \psi(\nu) = \LPLSL \exp[\frac{V}{\lambda} \sum_{k=1}^\infty \frac{(i\nu)^k}{k!} \SPL(k)
			+ O(\frac{V}{\lambda} \sum_{k=1}^\infty \frac{(i\nu)^k}{k!} E[\delta(B_1)])
\end{equation}
But since $\delta(j) \to 0$ for all $j$ we have,
\begin{equation}
\LPLSL \psi(\nu) = \LPLSL \exp[\frac{V}{\lambda} \sum_{k=1}^\infty \frac{(i\nu)^k}{k!} \SPL(k))
\end{equation}
An almost identical argument shows that the characteristic function of $\sum_{j=1}^Y \brand_j(\srand_j)$ converges to the same limit.  
	We have demonstrated (\ref{E:p_lim}).  (\ref{E:q_lim}) is demonstrated in an identical way.  To demonstrate (\ref{E:Fst_lim}) we simply compute the characteristic function of the pair $(\p, \q)$.  The arguments are almost identical to those we made in deriving (\ref{E:p_lim}) so we do not include them here.

\end{proof}

	We are finally ready to state and prove Theorems \ref{T:middle_tree} and \ref{T:top_tree}.  Their proofs are very similar so we prove only Theorem \ref{T:middle_tree}.

\begin{proof}[Theorem \ref{T:middle_tree}]
Set
\begin{gather}
\hat{\p} = \sum_{k=1}^Q X_k, \\ \notag
\hat{\q} = \sum_{k=1}^Q X_k^2.
\end{gather}
Let $\nu = (\nu_1, \nu_2)$.  We need to show
\begin{equation}
\LPLSL E[\exp[i\nu \cdot (\p, \q)]] = E[\exp[i \nu \cdot (\hat{\p}, \hat{\q})]].
\end{equation}
We have actually already done most of the work in the proof of Lemma \ref{L:pq_limits}.  The arguments in the proof of Lemma \ref{L:pq_limits} show
\begin{equation}
\LPLSL E[\exp[i\nu \cdot (\p, \q)]] = \exp[\frac{V}{\lambda} \sum_{k=1}^\infty \sum_{j=0}^k 
					\binom{k}{j} \frac{i^k\nu_1^j \nu_2^{k-j}}{k!} \SPL(2k-j)]
\end{equation}
A standard computation shows that this is exactly the value of $E[\exp[i \nu \cdot (\hat{\p}, \hat{\q})]]$.

\end{proof}